\documentclass[twoside,a4paper]{article}
\usepackage{amsmath,graphicx,amssymb,fancyhdr,amsthm}
\newtheorem{thm}{Theorem}[section]
\newtheorem{cor}[thm]{Corollary}

\newtheorem{prop}[thm]{Proposition}

\theoremstyle{definition}
\newtheorem{defn}[thm]{Definition}
\theoremstyle{remark}

\def\beq{\begin{eqnarray}}
\def\eeq{\end{eqnarray}}
\def\bsp{\begin{split}}
\def\esp{\end{split}}

\def\d{\mathrm{d}}
\def\diag{\mathrm{diag}}
\def\i{\mathrm{i}}

\newcommand{\mbold}[1]{\mbox{\boldmath{\ensuremath{#1}}}}

\begin{document}

\title{\Large\textbf{Spacetimes characterized by their scalar curvature invariants}}
\author{{\large\textbf{Alan Coley$^{1}$, Sigbj\o rn Hervik$^{2}$, Nicos Pelavas$^{1}$} }
 \vspace{0.3cm} \\
$^{1}$Department of Mathematics and Statistics,\\
Dalhousie University,
Halifax, Nova Scotia,\\
Canada B3H 3J5
\vspace{0.2cm}\\
$^{2}$Faculty of Science and Technology,\\
 University of Stavanger,\\  N-4036 Stavanger, Norway
\vspace{0.3cm} \\
\texttt{aac, pelavas@mathstat.dal.ca, sigbjorn.hervik@uis.no} }
\date{\today}
\maketitle
\pagestyle{fancy}
\fancyhead{} 
\fancyhead[EC]{A. Coley, S. Hervik and N. Pelavas}
\fancyhead[EL,OR]{\thepage}
\fancyhead[OC]{Spacetimes characterized by their scalar curvature invariants}
\fancyfoot{} 

\begin{abstract}

In this paper we determine the class of four-dimensional
Lorentzian manifolds that can be completely characterized by the
scalar polynomial curvature invariants constructed from the Riemann tensor
and its covariant derivatives. We introduce the notion of an
$\mathcal{I}$-non-degenerate spacetime metric, which implies that
the spacetime metric is locally determined by its curvature
invariants. By determining an appropriate set of projection
operators from the Riemann tensor and its covariant derivatives,
we are able to prove a number of results (both in the
algebraically general and in algebraically special cases) of when
a spacetime metric is $\mathcal{I}$-non-degenerate. This enables
us to prove our main theorem that a spacetime metric is either
$\mathcal{I}$-non-degenerate or a Kundt metric. Therefore, a
metric that is not characterized by its curvature invariants must
be of Kundt form. We then discuss the inverse question of what
properties of the underlying spacetime can be determined from a
given a set of scalar polynomial invariants, and some partial
results are presented. We also discuss the notions of
\emph{strong} and \emph{weak} non-degeneracy.

\end{abstract}

\section{Introduction}
In matters related to relativity and gravitational physics we are often interested in comparing various spacetime metrics. Often identical metrics (which, of course, would give identical physics) are given in different coordinates and will therefore be disguising their true equivalence. It is therefore of import to have an invariant way to distingush spacetime metrics. The perhaps easiest way of distinguishing metrics is to calculate (some of) their scalar polynomial curvature invariants due to the fact that inequivalent invariants implies inequivalent metrics. However, if their scalar polynomial invariants are the same, what conclusion can we draw about the (in)equivalence of the metrics? For example, if all such invariants are zero, can we say that the metric is flat? The answer to this question is known to be no, because all so-called VSI metrics have vanishing scalar invariants. Here, we will address the more general question: if two spacetimes have identical scalar polynomical curvature invariants, what can we say about these spacetimes? In particular, when do the invariants characterise the spacetime metric? 

For a spacetime $(\mathcal{M},g)$ with a set of scalar polynomial curvature invariants, there
are two conceivable ways in which the metric $g$ can be altered
such that the invariants remain the same. First, the metric can be
continuously deformed in such a way that the invariants remain
unchanged. This is what happens for the Kundt metrics for which we
have free functions which do not affect the curvature invariants.
Alternatively,  a discrete transformation of the metric can leave
the invariants the same. A simple example of when a discrete
transformation can give another metric with the same set of
invariants is the pair of metrics: \beq\label{weakly1}
\d s^2_1&=&\frac{1}{z^2}\left(\d x^2+\d y^2+\d z^2\right)-\d \tau^2, \\
\d s^2_2&=&\frac{1}{z^2}\left(-\d x^2+\d y^2+\d z^2\right)+\d \tau^2. \label{weakly2}
\eeq
One can straight-forwardly check that these metrics have identical invariants but are not diffeomorphic (over the reals). These discrete transformations are more difficult to deal with but the issue will be taken up in a later section.

Therefore, first we will  consider the continuous metric
deformations defined as follows.
\begin{defn}
\noindent For a spacetime $(\mathcal{M},g)$, a (one-parameter) \emph{metric deformation}, $\hat{g}_\tau$, $\tau\in [0,\epsilon)$, is a family of smooth metrics on $\mathcal{M}$ such that
\begin{enumerate}
\item{} $\hat{g}_\tau$ is continuous in $\tau$,
\item{} $\hat{g}_0=g$; and
\item{} $\hat{g}_\tau$ for $\tau>0$, is not diffeomorphic to $g$.
\end{enumerate}
\end{defn}
\noindent For any given spacetime $(\mathcal{M},g)$ we define the set of all scalar polynomial curvature invariants
\begin{equation}
\mathcal{I}\equiv\{R,R_{\mu\nu}R^{\mu\nu},C_{\mu\nu\alpha\beta}C^{\mu\nu\alpha\beta}, R_{\mu\nu\alpha\beta;\gamma}R^{\mu\nu\alpha\beta;\gamma}, R_{\mu\nu\alpha\beta;\gamma\delta}R^{\mu\nu\alpha\beta;\gamma\delta},\dots\} \,. \nonumber
\end{equation}
\noindent Therefore, we can consider the set of invariants as a function of the metric
and its derivatives. However, we are interested in to what extent,
or under what circumstances, this function has an inverse.
\begin{defn}
Consider a spacetime $(\mathcal{M},g)$ with a set of invariants
$\mathcal{I}$. Then, if there does not exist a metric deformation
of $g$ having the same set of invariants as $g$, then we will call
the set of invariants \emph{non-degenerate}. Furthermore, the
spacetime metric $g$, will be called
\emph{$\mathcal{I}$-non-degenerate}.
\end{defn}

This implies that for a metric which is
$\mathcal{I}$-non-degenerate the invariants characterize the
spacetime uniquely, at least locally, in the space of (Lorentzian)
metrics. This means that these metrics are characterized by their
curvature invariants and therefore we can distinguish such metrics
using their invariants. Since scalar curvature invariants are
manifestly diffeomorphism-invariant we can thereby avoid the
difficult issue whether a diffeomorphism exists connecting two
spacetimes.

\section{Main Theorems}
Let us first state our main theorems which will be proven in the
later sections. The theorems apply to four-dimensional (4D)
Lorentzian manifolds. Such spacetimes are characterized
algebraically by their Petrov \cite{Petrov,kramer} and Segre \cite
{HallBook,kramer} types or, equivalently, in terms of their Ricci,
Weyl (and Riemann) types \cite {coley,Milson,class}. The notation,
which essentially follows that of the cited references, is briefly
summarized in Appendix A. The proofs of these theorems, which are
investigated on a case by case basis in terms of the algebraic
type of the curvature tensors, are long and tedious and have
therefore been placed in later sections. Once all of the various
cases have been explored the theorems follow.

Furthermore, let us remark on the technical assumptions made in this paper. The following theorems hold on neighborhoods where the Riemann, Weyl and Segre types do not change. In the algebraically special cases we also need to assume that the algebraic type of the higher-derivative curvature tensors also do not change, up to the appropriate order. Most crucial is the definition of the curvature operators (see later) and in order for these to be well defined, the algebraic properties of the curvature tensors need to remain the same over a neighborhood.\footnote{Alternatively, we can assume that the spacetime is real analytic.} Henceforth, we will therefore assume that we consider an open neighborhood where the algebraic properties of the curvature tensors do not change, up to the appropriate order ($\leq 4$).

The first theorem deals with the algebraic classification of the curvature tensors, and the relation to the $\mathcal{I}$-non-degenerate metrics.

\begin{thm}[Algebraically general]\label{thm:alggen}
If a spacetime metric is of Ricci type $I$, Weyl type $I$, or Riemann type $I/G$,  the metric is $\mathcal{I}$-non-degenerate.
\end{thm}
This theorem indicates that the general metric is
$\mathcal{I}$-non-degenerate and thus the metric is determined by
its curvature invariants (at least locally, in the sense explained
above). In the above, by Riemann type $I/G$, we are referring to the existence of a frame in which components of boost weight +2 vanish for Riemann type $I$, and in type $G$ there does not exist a frame in which components with boost weight +2 or -2 vanish, in this case the Weyl and Ricci canonical frames are not aligned.
For the algebraically special spacetimes, we need to consider
covariant derivatives of the Riemann tensor.
\begin{thm}[Algebraically special]\label{thm:algspecial}
If the spacetime metric is algebraically special, but $\nabla R$, $\nabla^{(2)} R$, $\nabla^{(3)} R$, or $\nabla^{(4)} R$ is of type $I$ or more general, the metric is $\mathcal{I}$-non-degenerate.
\end{thm}
In terms of the boost weight decomposition, an algebraically
special metric has a Riemann tensor with zero positive boost
weight components. In general, type $I$ refers to the vanishing of
boost weight components +2 and higher (but not boost weight +1
components). For example, we often use the notation $(\nabla
R)_{b}=0$, $b\geq 2$ to denote a $\nabla R$ of type $I$ (but
$(\nabla R)_{1} \neq 0$). The above theorem indicates that if by
taking covariant derivatives of the Riemann tensor you acquire
positive boost weight components, then the metric is
$\mathcal{I}$-non-degenerate. The remaining metrics which \emph{do
not} acquire a positive boost weight component when taking
covariant derivatives, have a very special structure of their
curvature tensors. Indeed, such metrics must be very special
metrics:
\begin{thm}\label{thm:main}
Consider a spacetime metric. Then either,
\begin{enumerate}
\item{}the metric is $\mathcal{I}$-non-degenerate; or,
\item{}the metric is contained in the Kundt class.
\end{enumerate}
\end{thm}
This is a striking result because it tells us that metrics
not determined by their curvature invariants must be of
Kundt form. These Kundt metrics therefore correspond to degenerate
metrics in the sense that many such spacetimes can have identical
invariants.  The Kundt class is defined by those metrics admitting a
null vector $\ell$ that is geodesic, expansion-free, shear-free and twist-free (corresponding to the vanishing of the spin-coefficients $\kappa$, $\sigma$ and $\rho$; see also Appendix \ref{notation})
\begin{equation}
\begin{array}{cccc}
\ell^{\beta}\ell_{\alpha;\beta}=0 & \ell_{\alpha}^{\ ;\alpha}=0 & \ell^{(\alpha;\beta)}\ell_{\alpha;\beta}=0 & \ell^{[\alpha;\beta]}\ell_{\alpha;\beta}=0 \, .
\end{array}
\end{equation}
Any metric in the Kundt class can be written in the following canonical form \cite{CSI,coley}:
\beq
\d s^2=2\d u\left[\d v +H(v,u,x^k)\d
u+W_{i}(v,u,x^k)\d x^i\right]+g_{ij}(u,x^k)\d x^i\d x^j.
\label{HKundt}\eeq

For spacetimes with constant curvature invariants (CSI)
Theorem \ref{thm:main} has an important consequence. For CSI metrics,
$\mathcal{I}$-non-degenerate implies that the spacetime is
curvature homogeneous to all orders; hence, an important corollary
is a proof of the CSI-Kundt conjecture \cite{CSI}:

\begin{cor}[CSI spacetimes]
Consider a 4-dimensional spacetime having all constant curvature invariants (CSI). Then either,
\begin{enumerate}
\item{} the spacetime is locally homogeneous; or,
\item{} a subclass of the Kundt spacetimes.
\end{enumerate}
\end{cor}

These theorems imply that the Kundt spacetimes play a pivotal role in the
question of which metrics are $\mathcal{I}$-non-degenerate.  Indeed, the Kundt
metrics are the \emph{only metrics not determined by their curvature invariants}
(in the sense explained above).



In fact, we can be somewhat more precise since only a subclass of
the Kundt spacetimes have these exceptional properties. In the
analysis (described below) it is found that a Kundt metric is
$\mathcal{I}$-non-degenerate if the metric functions $W_i(v,u,x^k)$ in
the canonical ({\it kinematic}) Kundt null frame are non-linear in
$v$ (i.e., $W_{i,vv} \neq 0$). Hence the exceptional spacetimes are
the {\it aligned algebraically special type-$II$}-Kundt spacetimes
or, in short (and consistent with the terminology of the above
theorem) {\it degenerate} Kundt spacetimes, in which there exists
a common null frame in which the geodesic, expansion-free,
shear-free and twist-free null vector $\ell$ is also the null
vector in which all positive boost weight terms of the Riemann
tensor are zero (i.e., the kinematic Kundt frame and the Riemann
type $II$ aligned null frame are {\it aligned}). We note that the
important Kundt-CSI and vanishing scalar invariant (VSI) spacetimes are degenerate Kundt spacetimes \cite{4DVSI,Higher,nDVSI,CSI}.


\section{Curvature operators and  curvature projectors}
In order to prove the main theorems we need to introduce some mathematical tools. These tools, although they are very simple, are extremely useful and powerful in proving these theorems.

A curvature operator, ${\sf T}$, is a tensor considered as a (pointwise) linear operator
\[ {\sf T}:~V\mapsto V, \]
for some vector space $V$, constructed from the Riemann tensor, its covariant derivatives, and the curvature invariants.

The archetypical example of a curvature operator is obtained by
raising one index of the Ricci tensor.  The Ricci operator is
consequently a mapping of the tangent space $T_p\mathcal{M}$ into
itself:
\[ {\sf R}\equiv(R^{\mu}_{~\nu}):~T_p\mathcal{M}\mapsto T_p\mathcal{M}. \]
Another example of a curvature operator is the Weyl tensor, considered as an operator, ${\sf C}\equiv (C^{\alpha\beta}_{\phantom{\alpha\beta}\mu\nu}$), mapping bivectors onto bivectors.

For a curvature operator, ${\sf T}$, consider an eigenvector ${\sf
v}$ with eigenvalue $\lambda$; i.e., ${\sf T}{\sf v}=\lambda{\sf
v}$. If $d=\mathrm{dim}(V)$ and $n$ is the dimension of the
spacetime, then the eigenvalues of ${\sf T}$ are $GL(d)$ invariant. Since
the Lorentz transformations, $O(1,n-1)$, will act via a
representation $\Gamma\subset GL(d)$ on ${\sf T}$,
 \emph{the eigenvalues of a curvature operator is an $O(1,n-1)$-invariant curvature scalar}.
Therefore, curvature operators naturally provide us with a set of
curvature invariants (not necessarily polynomial invariants)
corresponding to the set of distinct eigenvalues: $\{\lambda_A
\}$. Furthermore, the set of eigenvalues are uniquely determined
by the polynomial invariants of ${\sf T}$ via its characteristic
equation. The characteristic equation, when solved, gives us the
set of eigenvalues, and hence these are consequently determined by
the invariants.
{\footnote { Note that the 'corresponding eigenvalues'
of the operators constructed from the covariant derivatives of the Riemann tensor
are also related to scalar curvature invariants built from
covariant derivatives.}}

We can now define a number of associated curvature operators. For
example,  for an eigenvector ${\sf v}_A$ so that ${\sf T}{\sf
v}_A=\lambda_A{\sf v}_{A}$, we can construct the annihilator
operator:
\[ {\sf P}_A\equiv ({\sf T}-\lambda_{A}{\sf 1}).\]
Considering the Jordan block form of ${\sf T}$, the eigenvalue ${\lambda_A}$ corresponds to a set of Jordan blocks. These blocks are of the form:
\[ {\sf B}_A=\begin{bmatrix}
\lambda_A & 0 & 0& \cdots & 0 \\
1 & \lambda_A & 0& \ddots  & \vdots \\
0      & 1 &\lambda_A& \ddots & 0 \\
\vdots & \ddots     &\ddots& \ddots & 0 \\
0    &     \hdots &0   &  1      & \lambda_A
\end{bmatrix}.\]
There might be several such blocks corresponding to an eigenvalue
$\lambda_A$; however, they are all such that $({\sf
B}_A-\lambda_A{\sf 1})$ is nilpotent and hence there exists an
$n_{A}\in \mathbb{N}$ such that  ${\sf P}_A^{n_A}$ annihilates the
whole vector space associated to the eigenvalue $\lambda_A$.

This implies that we can define a set of operators $\widetilde{\bot}_A$ with eigenvalues $0$ or $1$ by considering the products 
\[ \prod_{B\neq A}{\sf P}^{n_B}_B=\Lambda_A\widetilde{\bot}_A,\] 
where $\Lambda_A=\prod_{B\neq A}(\lambda_A-\lambda_B)^{n_B}\neq 0$ 
(as long as $\lambda_B\neq \lambda_A$ for all $B$). Furthermore, we can now define
\[ \bot_A\equiv {\sf 1}-\left({\sf 1}-\widetilde{\bot}_A\right)^{n_A}  \]
where $\bot_A$
is a \emph{curvature projector}. The set of all such curvature
projectors obeys: \beq {\sf 1}=\bot_1+\bot_2+\cdots+\bot_A+\cdots,
\quad \bot_A\bot_B=\delta_{AB}\bot_A. \eeq We can use these
curvature projectors to decompose the operator ${\sf T}$: \beq
{\sf T}={\sf N}+\sum_A\lambda_A\bot_A. \label{decomp} \eeq The
operator ${\sf N}$ therefore contains  all the information not
encapsulated in the eigenvalues $\lambda_A$. From the Jordan form
we can see that ${\sf N}$ is nilpotent; i.e., there exists an
$n\in\mathbb{N}$ such that ${\sf N}^n={\sf 0}$. In particular, if
${\sf N}\neq 0$, then ${\sf N}$ is a negative/positive
boost weight operator which can be used to lower/raise the
boost weight of a tensor.

Considering the Ricci operator, or the Weyl operator, we can show
that (where the type refers to either Ricci type or Weyl type):
\begin{itemize}
\item{} Type I: ${\sf N}={\sf 0}$, $\lambda_A\neq 0$.
\item{} Type D: ${\sf N}={\sf 0}$, $\lambda_A\neq 0$.
\item{} Type II: ${\sf N}^3={\sf 0}$, $\lambda_A\neq 0$.
\item{} Type III: ${\sf N}^3={\sf 0}$, $\lambda_A=0$.
\item{} Type N: ${\sf N}^2={\sf 0}$, $\lambda_A=0$.
\item{} Type O: ${\sf N}={\sf 0}$, $\lambda_A=0$.
\end{itemize}

Consider a curvature projector $\bot: T_p\mathcal{M}\mapsto T_p\mathcal{M}$. Then, for a Lorentzian spacetime there are 4 categories:
\begin{enumerate}
\item{} Timelike: For all $v^{\mu}\in T_p\mathcal{M}$, $v_{\nu}(\bot)^{\nu}_{~\mu}v^{\mu}\leq 0$.
\item{} Null: For all $v^{\mu}\in T_p\mathcal{M}$, $v_{\nu}(\bot)^{\nu}_{~\mu}v^{\mu}= 0$.
\item{} Spacelike: For all $v^{\mu}\in T_p\mathcal{M}$, $v_{\nu}(\bot)^{\nu}_{~\mu}v^{\mu}\geq 0$.
\item{} None of the above.
\end{enumerate}

In the following, we shall consider a complete set of curvature
projectors: $\bot_A: T_p\mathcal{M}\mapsto T_p\mathcal{M}$. These
projectors can be of any of the aforementioned categories and we
are going to use the Segre-like notation to characterize the set
with a comma separating time and space. For example, $\{1,111\}$
means we have 4 projectors: one timelike, and three spacelike. A
bracket indicates that the image of the projectors are of
dimension 2 or higher; e.g., $\{(1,1)11\}$ means that we have two
spacelike operators, and one with a 2 dimensional image. If there
is a null projector, we automatically have a second null
projector. Given an NP frame $\{
\ell_{\mu},n_{\mu},m_{\mu},\bar{m}_{\mu}\}$, then a null-projector
can typically be:
\[ (\bot_1)^{\mu}_{~\nu}=-\ell^{\mu}n_{\nu}. \]
Note that $\bot_1^2=\bot_1$, but it is not symmetric. Therefore, acting from the left and right gives two different operators. Indeed, defining
\[ (\bot_2)^{\mu}_{~\nu}\equiv g_{\nu\alpha}g^{\mu\beta}(\bot_1)^{\alpha}_{~\beta}, \]
we get a second null-projector being orthogonal to $\bot_1$. The
existence of null-projectors enables us to pick out certain null
directions; however, note that the null-operators, with respect to
the aforementioned Newman-Penrose (NP) frame, are of boost weight
0 and so cannot be used to lower/raise the boost weights. In
particular, considering the combination $\bot_1+\bot_2$ we see
that the existence of null-projectors implies the existence of
projectors of type $\{(1,1)(11)\}$.

The existence of curvature projectors is important due to the following result:
\begin{thm}\label{thm:lochom}
Consider a spacetime metric and assume that there exist curvature projectors of type $\{1,111\}$, $\{1,1(11)\}$ or $\{1,(111)\}$. Then the spacetime is \emph{$\mathcal{I}$-non-degenerate}.
\end{thm}
\begin{proof}
Consider first the case $\{1,111\}$. For any given curvature tensor, $R_{\alpha\beta...\delta}$, we can  construct the curvature tensor
\[ R[ij...k]_{\alpha\beta...\delta}\equiv R_{\mu\nu...\lambda}(\bot_i)^{\mu}_{~\alpha}(\bot_j)^{\nu}_{~\beta}...(\bot_k)^{\lambda}_{~\delta} \,. \]
This enables us to consider the curvature invariant $
R[ij...k]_{\alpha\beta...\delta} R[ij...k]^{\alpha\beta...\delta}$
which is, up to a constant factor, the square of the component
$R_{ij...k}$. This implies that it is determined by the invariant
(up to a sign) and we get that \emph{the spacetime is
$\mathcal{I}$-non-degenerate.}

Consider now the case $\{1,1(11)\}$. We note that in this case we cannot isolate all components of the curvature tensors. However, we can uniquely define tensors $r^{(A)}_{IJ...K}$, $I,J,...=3,4$ by contractions with $\bot_i$. The curvature invariants will now be $SO(2)$-invariant polynomials in the components of $r^{(A)}_{IJ...K}$. Hence, since $SO(2)$ is compact, the polynomials will separate the $SO(2)$ orbits. Hence, by a similar proof as in \cite{prufer} we get that \emph{the spacetime is $\mathcal{I}$-non-degenerate.}

Lastly, consider the case $\{1,(111)\}$. In this case we can define tensors $r^{(A)}_{IJ...K}$, $I,J,...=2,3,4$ by contractions with $\bot_i$. The curvature invariants will be $SO(3)$-invariant which is again compact. Hence, using a similar argument as in \cite{prufer} we get that \emph{the spacetime is $\mathcal{I}$-non-degenerate.}
\end{proof}

\section{Riemann type $I/G$}
Let us consider first the case where the Riemann tensor is of type
I or $G$. This corresponds to the three cases: Ricci type $I$, Weyl
(Petrov) type $I$, and Ricci and Weyl canonical frames not aligned.
We shall consider these in turn.
\subsection{Ricci type $I$}
This case consists of the following Segre types: $\{1,111\}$, $\{1,1(11)\}$ $\{1,(111)\}$, $\{z\bar{z}11\}$,  $\{z\bar{z}(11)\}$.
\subsubsection{Segre type $\{1,111\}$:} Here the eigenvalues of the Ricci operator are all distinct and we can diagonalize the Ricci operator:
\[ {\sf R}=\diag(\lambda_1,\lambda_2,\lambda_3,\lambda_4).\]
It now follows from Theorem \ref{thm:lochom} that  \emph{the spacetime is $\mathcal{I}$-non-degenerate.}

Indeed, to determine the spacetime it is sufficient to consider $R_{\mu\nu;\alpha}$.
Choosing an orthonormal frame, ${\bf e}_i$, aligned with the eigendirections of ${\sf R}$:
\[ R_{ij;k}=\lambda_{i,k}g_{ij}+(\lambda_{i}-\lambda_j)\Gamma_{ijk},\]
where $ \Gamma_{ijk}$ are the connection coefficients, we find
that all connection coefficients must be determined by the
curvature invariants.
\subsubsection{Segre type $\{1,1(11)\}$:} This is the special case of above where we have
$\lambda_3=\lambda_4$. Using Theorem \ref{thm:lochom} \emph{the spacetime is
$\mathcal{I}$-non-degenerate.}
\subsubsection{Segre type $\{1,(111)\}$:}  Here we have $\lambda_2=\lambda_3=\lambda_4$ and from Theorem \ref{thm:lochom} we have that \emph{the spacetime is $\mathcal{I}$-non-degenerate.}

\subsubsection{Segre type $\{z\bar{z}11\}$ and $\{z\bar{z}(11)\}$:}
In this case, the Ricci operator has two complex conjugate eigenvalues. We can always find an orthonormal frame $\{{\bf e}_{i}\}$, so that the Ricci operator takes the form
\[ {\sf R}=\begin{bmatrix}
a & b & 0 & 0 \\
-b& a & 0 & 0 \\
0 & 0 & \lambda_3 & 0 \\
0 & 0 & 0 &\lambda_4
\end{bmatrix}.\]
We can now consider the complex transformation mapping the basis vectors ${\bf e}_0$ and ${\bf e}_1$ onto the eigenvectors ${\bf v}_0$ and ${\bf v}_1$:
\beq
{\bf v}_0= \tfrac{1}{\sqrt{2}}\left({\bf e}_0+\i{\bf e}_1\right), \qquad
{\bf v}_1= \tfrac{1}{\sqrt{2}}\left(\i{\bf e}_0+{\bf e}_1\right),
\eeq
with inverse
\beq
{\bf e}_0 = \tfrac{1}{\sqrt{2}}\left({\bf v}_0-\i{\bf v}_1\right), \qquad
{\bf e}_1 = \tfrac{1}{\sqrt{2}}\left(-\i{\bf v}_0+{\bf v}_1\right).
\eeq
We note that ${\bf v}_0\cdot{\bf v}_0=-1$, ${\bf v}_1\cdot{\bf v}_1=1$, ${\bf v}_0\cdot{\bf v}_1=0$ and so the set $\{{\bf v}_0,{\bf v}_1,{\bf e}_2,{\bf e}_3\}$ can be considered as an orthonormal frame. In this frame the Ricci operator becomes diagonal:
\[ {\sf R}=\diag(\lambda_1,\bar{\lambda}_1,\lambda_2,\lambda_3).\]
Therefore, we have a set of curvature projectors of the form
$\{1,111\}$ or $\{1,1(11)\}$ and  we can use Theorem
\ref{thm:lochom}. The only difference is that the invariants
associated to the complex frame can now be complex; however, the
result still stands. Using the inverse transformation, which
induces a transformation between the invariants from the complex
frame to the real frame, we obtain the curvature components of the
real frame. Therefore we can conclude that \emph{the spacetime is
$\mathcal{I}$-non-degenerate.}

\subsection{Weyl type $I$ (Petrov type $I$)}
For the Weyl tensor any non-trivial isotropy would make it
algebraically special. The isotropy group of the Weyl tensor is
the subgroup of the Lorentz group whose action on the Weyl tensor
leaves it invariant; for example a Petrov type D Weyl tensor has a
boost-spin isotropy group. So for the Weyl tensor to be of type $I$
requires that the isotropy group is trivial. We therefore expect
that we will be able to determine a unique frame using the
curvature invariants.

We use the bivector formalism and write the Weyl tensor, $C_{\alpha\beta\mu\nu}$, as an operator in 6-dimensional bivector space, ${\sf C}=(C^A_{~B})$. Using the following index convention:
\[ [23]\leftrightarrow 1,~[31]\leftrightarrow 2, ~[12]\leftrightarrow 3, ~[10]\leftrightarrow 4,~[20]\leftrightarrow 5,~[30]\leftrightarrow 6, \]
a type I Weyl tensor can always be put into the following
canonical form \cite{HallBook}: 
\beq {\sf C}=\begin{bmatrix}
a_1 & 0 & 0 & b_1 & 0 & 0 \\
0 & a_2 & 0 & 0 & b_2 & 0 \\
0 & 0 & a_3 & 0 & 0 & b_3 \\
-b_1 & 0 & 0 & a_1 & 0 & 0 \\
0 & -b_2 & 0 & 0 & a_2 & 0 \\
0 & 0 & -b_3 & 0 & 0 & a_3
\end{bmatrix}\label{Coperator}
\eeq where $\sum_ia_i=\sum_ib_i=0$ and not all of the $a_i$, $b_i$
are zero.

First we note that the eigenvalues of ${\sf C}$ are $a_i\pm \i
b_i$. As explained above, $a_i$ and $b_i$ are uniquely determined
by the zeroth order Weyl invariants. The eigenbivectors are
$F^A=\delta^A_i\pm \i\delta^A_{3+i}$. We can therefore construct
annihilator operators, $({\sf C}-\lambda{\sf 1})$, and projection operators as before (the only difference is that
${\sf C}$ is 6-dimensional). The eigenbivectors correspond to
(complex) antisymmetric tensors. For example, consider the
eigenbivector with eigenvalue $a_1+\i b_1$:
\[ F=\frac 12F_{\mu\nu}{\mbold\omega}^\mu\wedge{\mbold\omega}^\nu={\mbold\omega}^2\wedge{\mbold\omega}^3-\i {\mbold\omega}^1\wedge{\mbold\omega}^0.\]
Hence, from this we can construct an operator
\beq
{\sf P}_1=(F^{\mu}_{~\alpha}\bar{F}^{\alpha}_{~\nu})=
\begin{bmatrix}
1 & 0 & 0 & 0 \\ 0 & 1 & 0 & 0 \\ 0 & 0 & -1 & 0 \\ 0 & 0 & 0 & -1
\end{bmatrix}\, .
\eeq For the other eigenbivectors we then get (analogously):
\[ {\sf P}_2=\diag(1,-1,1,-1), \quad {\sf P}_3=\diag(1,-1,-1,1).\]
Thus the linear set $\{{\sf 1},{\sf P}_i\}$ span all diagonal
matrices; in particular, we can construct the projection
operators: \beq
\bot_1&=& \tfrac 14({\sf 1}+{\sf P}_1+{\sf P}_2+{\sf P}_3), \nonumber \\
\bot_2&=& \tfrac 14({\sf 1}+{\sf P}_1-{\sf P}_2-{\sf P}_3), \nonumber \\
\bot_3&=& \tfrac 14({\sf 1}-{\sf P}_1+{\sf P}_2-{\sf P}_3), \nonumber \\
\bot_4&=& \tfrac 14({\sf 1}-{\sf P}_1-{\sf P}_2+{\sf P}_3) \,.
\nonumber \eeq It is clear that we will get 3 operators, ${\sf
P}_i$, as long as the 3 sets of complex eigenvalues,
$\lambda_i=a_i+\i b_i$, are all different. Since
$\sum_i\lambda_i=0$, this can only fail when: \beq
\lambda_1=\lambda_2, ~\lambda_3=-2\lambda_1, \nonumber \\
\lambda_1=\lambda_2=\lambda_3=0. \nonumber
\eeq
The first of these is actually Weyl type $D$, while the latter is Weyl type $O$; hence, these are excluded by assumption.

Therefore, we can conclude that as long as the Weyl type is $I$ (and
not simpler), we can define 4 projection operators of type
$\{1,111\}$. Therefore, from Theorem \ref{thm:lochom}, \emph{the
spacetime is $\mathcal{I}$-non-degenerate}.

At this stage we wish to remark on a certain subtlety in the
choice of eigenvectors.  From the Weyl tensor we can actually only
determine the product $F_{\mu\nu}F_{\alpha\beta}$. Therefore, we
can only construct the ``square'' ${\sf P}_1\otimes{\sf P}_1$. So
in order to get the operator ${\sf P}_1$ there is an ambiguity in
the choice of sign. Regarding the question of
$\mathcal{I}$-non-degeneracy as defined above this has no
consequence; however, it may have an effect on discrete changes to
the metric. This sign ambiguity results in a
permutation of the axes; essentially, we don't know which axis
corresponds to time. We will get back to this issue later but note
that this phenomenon will recur in several cases below.

\subsection{Ricci and Weyl canonical frames not aligned}

Consider now the case where both the Ricci tensor and the Weyl
tensor are algebraically special but where there does not exist a
null-frame such that both the Ricci tensor and the Weyl tensor has
only non-positive boost weights.

First, assume the Weyl type is $D$ and choose the Weyl canonical
frame. For Weyl type $D$ the Weyl operator is of the form of eq.
(\ref{Coperator}) with $\lambda_1=\lambda_2,
~\lambda_3=-2\lambda_1$. This immediately implies we have
projection operators of type $\{(1,1)(11)\}$.

In the Weyl canonical frame, the Ricci tensor must have  both
positive \emph{and} negative boost weight components (or else
there would exist a frame where they are aligned). Now, by
symmetry, we can consider three cases for the Ricci tensor (see Appendix \ref{notation} for $(R)_{b}$ notation):
\begin{enumerate}
\item{} $(R)_{+2}\neq 0$, $(R)_{-2}\neq 0$: Here, we use the
$(1,1)$-projection operator and we get a reduced Ricci operator of
the form (in the $\{\ell,n\}$ frame): \beq \widetilde{{\sf
R}}=\begin{bmatrix} a & b \\ c & a \end{bmatrix}, \quad bc\neq 0.
\eeq This gives two distinct eigenvalues $\lambda=a\pm\sqrt{bc}$,
and hence, two additional projection operators. This case
therefore reduces to the case $\{1,1(11)\}$ or $\{z\bar{z}(11)\}$
presented earlier. This spacetime is therefore
\emph{$\mathcal{I}$-non-degenerate.} \item{} $(R)_{+2}=0$,
$(R)_{+1}\neq 0$, $(R)_{-2}\neq 0$: For this case we note that the
square $R^{\mu}_{~\alpha}R^{\alpha}_{~\nu}$ necessarily has boost
weight +2 components. Therefore, using either
$R^{\mu}_{~\alpha}R^{\alpha}_{~\nu}$ or
$R^{\mu}_{~\alpha}R^{\alpha}_{~\nu}+qR^{\mu}_{~\nu}$, where $q$ is
a parameter, we can use the results of the previous paragraph.
This case is consequently \emph{$\mathcal{I}$-non-degenerate}.
\item{} $(R)_{\pm 2}=0$, $(R)_{+1}\neq 0$, $(R)_{-1}\neq 0$: Here
we consider $R^{\mu}_{~\alpha}R^{\alpha}_{~\nu}$ which necessarily
has non-zero boost weight $-2$ and $+2$ components. This case is
therefore also \emph{$\mathcal{I}$-non-degenerate}.
\end{enumerate}

Assume now Weyl type $N$ and choose the Weyl canonical frame such
that $C=(C)_{-2}$. In this frame, either $(R)_{+1}$ or $(R)_{+2}$
is non-zero. If $(R)_{+2}$ is zero, we replace the Ricci tensor
$R_{\mu\nu}$ with the square $R^{\mu}_{~\alpha}R^{\alpha}_{~\nu}$
in what follows. Therefore, assume that $(R)_{+2}\neq 0$. Consider
now the operator
\[ (T^\mu_{~\nu})\equiv C^\mu_{~\alpha\nu\beta}R^{\alpha\beta}\]
Under the above assumptions, this operator can be used to get projectors of type $\{(1,1)11\}$. Indeed, these projectors are aligned with the Weyl canonical frame. We can now use one of the spacelike projectors, $\bot_3$ (say), to construct the symmetric operator:
\[ (\hat{T}^\mu_{~\nu})\equiv C^{\mu\alpha}_{\phantom{\mu\alpha}\nu\beta}(\bot_3)^{\beta}_{~\alpha}+qR^{\mu}_{~\nu},\]
where $q$ is a parameter. We can use this operator to construct
the remaining projection operators so that we have a set
$\{1,111\}$. This case is therefore
\emph{$\mathcal{I}$-non-degenerate.}

For Weyl type $II$ we can decompose the Weyl tensor:
\[ {\sf C}={\sf N}+\sum_A\lambda_A\bot_A,\]
where the operator ${\sf N}$ is a ``Weyl'' operator of type $N$ while the piece $\sum_A\lambda_A\bot_A$ is a ``Weyl'' operator of type $D$. By assumption, the Ricci tensor is not aligned with the Weyl canonical frame; therefore, using the above results, this case is also \emph{$\mathcal{I}$-non-degenerate}.

Lastly, for Weyl type $III$, we can consider the square ${\sf C}^2$ which is a Weyl operator of type $N$. The above results imply that this case is \emph{$\mathcal{I}$-non-degenerate}.

\subsubsection{Summary} Therefore, we have shown that: \emph{If a
4-dimensional spacetime $(\mathcal{M},g)$ is either Ricci type $I$,
Weyl type $I$ or Riemann type $I/G$, then it is
$\mathcal{I}$-non-degenerate. }

\section{Algebraically special cases}

For the algebraically special cases the Riemann tensor itself does
not give enough information to provide us with all the required
projection operators. Indeed, in the algebraically special cases
it is also necessary to calculate the covariant derivatives. The
strategy is as follows: we will consider the two cases of Weyl
type $D$ and $N$ in detail. The second Bianchi identity will not
be imposed at this time because we aim to use these results on
more general tensors with the same symmetries, not necessarily the
Weyl tensor itself. Weyl type $II$ and $III$ will now follow from
these computations and Weyl type $O$ will be treated last.

We should also point out that for any symmetric tensor $S_{\mu\nu}$ we can always construct a Weyl-like tensor with the same symmetries as the Weyl tensor. If $S_{\mu\nu}$ is the trace-free Ricci tensor, the corresponding Weyl-like tensor is the so-called Pleba\'{n}ski tensor. Explicitly, given the trace-free part of $S_{\mu\nu}$, denoted $\hat{S}_{\mu\nu}$, the Pleba\'{n}ski tensor is given by
\beq
W^{\alpha\beta}_{\phantom{\alpha\beta}\mu\nu}\equiv \hat{S}^{[\alpha}_{~[\mu}\hat{S}^{\beta]}_{~\nu]}+\delta^{[\alpha}_{~[\mu} \hat{S}_{\nu]\gamma}\hat{S}^{\beta]\gamma} - \frac 16 \delta^{[\alpha}_{~[\mu}\delta^{\beta]}_{~\nu]}\hat{S}_{\gamma\epsilon}\hat{S}^{\gamma\epsilon}.
\eeq
Therefore, to any symmetric tensor there is an associated ``Pleba\'{n}ski'' tensor.

Henceforth we are going to use the NP-formalism where we introduce a null frame
$\{\ell,n,m,\bar{m}\}$.  (We will use the notation of \cite{kramer}; also see Appendix
A).  In order to get the desired results we introduce the canonical frames for the
various algebraic types.  For the Weyl tensor, $C$, this means that we express its
components in terms of the Weyl scalars $\Psi_i$.  Then using the NP-connection
coefficients, we can express the covariant derivative $\nabla C$ in terms of $\Psi_i$
and the connection coefficients. At this stage it is useful not to assume anything about the connection $\nabla$ (i.e., the tensor $C$ need not be the Weyl tensor of the connection). The advantage of this is that we can utilise the full formalism of projection operators without worrying about the compatibility of the Weyl tensor and the connection. Furthermore, the results obtained here will therefore be more general than what is indicated. These expressions are then utilised
to obtain the required results for the curvature tensor. Another important thing to note is that when taking
covariant derivatives, some of the components have terms which are partial derivatives of
$\Psi_i$, while other terms are algebraic in $\Psi_i$ and $\Gamma^i_{jk}$.  These
algebraic terms are most useful simply because they give algebraic relations
rather than differential ones.

\subsection{Weyl (Petrov) type $D$}
We choose the canonical frame for which $\Psi_2\neq 0$. From the
Weyl operator ${\sf C}$ we can construct projectors of type
$\{(1,1)(11)\}$ where the $(1,1)$-projector projects onto the
$\ell-n$-plane, while the $(11)$-projector projects onto the
$m-\bar{m}$-plane. In the following let us use the indices $a,b,..
$ for projections onto the $\ell-n$ plane and the indices
$i,j,...$ for projections onto the $m-\bar{m}$ plane.

Calculating $\nabla C$ we get the boost weight decomposition
\[ \nabla C=(\nabla C)_{+2}+(\nabla C)_{+1}+(\nabla C)_{0}+(\nabla C)_{-1}+(\nabla C)_{-2}.\]
The key observation  is that the positive boost weight components vanish if and only if $\ell^{\mu}\nabla_{\mu}\Psi_2=0$ and $\kappa=\sigma=\rho=0$. Therefore, the idea is to define the appropriate operators so that we can isolate the necessary components.

Consider the (projected) tensor:
\[T_{iab}\equiv C^j_{~ij(a;b)}\]
This tensor has the following structure,
\[T_{iab}=v_in_an_b+t_i(\ell_an_b+n_a\ell_b)+w_i\ell_a\ell_b,\]
where
\beq
v_i&\equiv & 3(\bar{\Psi}_2\bar{\kappa}m_i+\Psi_2{\kappa}\bar{m}_i), \\
t_i&\equiv & -\tfrac 32 (\pi{\Psi}_2-\bar{\tau}\bar{\Psi}_2)m_i-\tfrac 32 (\bar{\pi}\bar{\Psi}_2-\tau{\Psi}_2)\bar{m}_i, \\
w_i&\equiv & -3({\Psi}_2{\nu}m_i+\bar{\Psi}_2\bar{\nu}\bar{m}_i). \\
\eeq

Furthermore, define the trace-free tensor $\widehat{T}_{iab}\equiv T_{iab}+(1/2)(\ell_an_b+n_a\ell_b)T_{i~c}^{~c}$, and then the tensor
\[ S_{abcd}=\widehat{T}^i_{~ab}\widehat{T}_{icd}.\]
This tensor can be considered as an operator ${\sf S}=(S^A_{~B})$
mapping symmetric trace-free tensors onto symmetric trace-free
tensors. For simplicity, let us also consider the trace-free part
of $S_{abcd}$ so that
\[ \widehat{S}_{abcd}=v^iv_in_an_bn_cn_d+w^iw_i\ell_a\ell_b\ell_c\ell_d.\]
Consider the trace-free tensor $M^{ab}=xn^an^b+y\ell^a\ell^b$. The
operator $\widehat{{\sf S}}$ has eigenvalues $\lambda=\pm |v||w|, ~0$.
Therefore, as long as both $v^i$ and $w^j$ are non-zero, there are
three distinct eigenvalues. Assuming $\lambda\neq 0$, $M^{ab}$ is an
eigentensor if $x=|v|$ and $y=|w|$. In this case we can consider
the curvature projectors (up to scaling), $M^{ab}M_{cd}$. The
eigentensor $M^{ab}$ can again be considered as an operator ${\sf
M}=(M^a_{~b})$ mapping vectors onto vectors. The eigenvalues of
${\sf M}$ are $\lambda=\pm i|v||w|$; hence, this reduces to the
case of two complex eigenvalues.

We note that $v^iv_i=0$ if and only if $\kappa=0$. Furthermore, if
either of $|v|$ or $|w|$ is non-zero we can assume, by using the
discrete symmetry defined later by  eq. (\ref{dsym}), that
$w^iw_i\neq 0$. Therefore, $\kappa\neq 0$ (so that $|v|\neq 0$
also)  implies that the spacetime is
\emph{$\mathcal{I}$-non-degenerate}.

Therefore, assume  $\kappa=0$ and consider the symmetric tensor
\[
Q^a_{~b}=C^{ijka;l}C_{ijkb;l}.
\]
The trace-free part of this tensor is
\[ \widehat{Q}_{ab}\propto |\Psi_2|^2\left(|\sigma|^2+|\rho|^2\right)n_an_b+|\Psi_2|^2\left(|\lambda|^2+|\mu|^2\right)\ell_a\ell_b.\]
If $\left(|\sigma|^2+|\rho|^2\right)\left(|\lambda|^2+|\mu|^2\right)\neq 0$, then this tensor is of type $I$. So from the Ricci type $I$ analysis, this would imply that the spacetime is \emph{$\mathcal{I}$-non-degenerate}.

Let us next consider the non-aligned case where $\kappa=0$,
$\lambda=\mu=0$ and $|\sigma|^2+|\rho|^2\neq 0$. We can now
consider the mixed tensor:
 \[ \widehat{Q}_{ab}+\widehat{S}_{abcd}\widehat{Q}^{cd}.\]
This tensor is of type I if $(w^iw_i)\left(|\sigma|^2+|\rho|^2\right)\neq 0$ and consequently,
the spacetime is \emph{$\mathcal{I}$-non-degenerate}.

Assume now that $w^i=0$, $\kappa=0$, for which we still have an
unused discrete symmetry (eq.(\ref{dsym})). If
$\left(|\sigma|^2+|\rho|^2\right)\left(|\lambda|^2+|\mu|^2\right)=0$
we can therefore assume that $\rho=\sigma=0$. This spacetime is
thus \emph{Kundt}.

Lastly, consider the case when $w^iw_i\neq 0$, $\kappa=\rho=\sigma=0$. This automatically implies that the spacetime is \emph{Kundt}.

Let us also consider the differential $\ell^{\mu}\nabla_{\mu}\Psi_2$ which in general
(not assuming the Bianchi identities are satisfied) also contributes to $(\nabla
C)_{+1}$.  We note that the Weyl invariant $I$, for a Weyl type $D$ tensor, is given by
$I=3\Psi_2^2$.  Therefore, we can consider the curvature tensor defined by the gradient
$\nabla_{\mu}I=6\Psi_2\Psi_{2,\mu}$.  We can now use the $(1,1)$-projector and project
this gradient onto the $\ell-n$-plane:  $x_a\equiv\nabla_{a}I$.  
\begin{enumerate}
\item{} If $x_a$ is either time-like or space-like (and consequently
$\ell^{\mu}\nabla_{\mu}\Psi_2\neq 0$), we can construct another curvature projector (by
considering the operator $x^ax_b$) so that we have a set $\{1,1(11)\}$.  Therefore, this
case is \emph{$\mathcal{I}$-non-degenerate}.  
\item{} If $x_a$ is null or zero, then
either $n^ax_a=0$ or $\ell^ax_a=0$.  If $\ell^ax_a=0$ then
$\ell^{\mu}\nabla_{\mu}\Psi_2=0$ and does not contribute to positive boost weight
components.  Assume therefore that $\ell^ax_a\neq 0$, which implies that $x_a\propto
n_a$.  If $\nu=\lambda=\mu=0$, we can use the discrete symmetry eq.(\ref{dsym}) so that
$\ell^{\mu}\nabla_{\mu}\Psi_2=0$.

If any of $\nu$, $\lambda$ or $\mu$ is non-zero, then $(\nabla C)_{b<0}$ is
non-zero.  Hence, by contracting with $x_a$, we can straight-forwardly construct another
projection operator so that we get a set $\{1,1(11)\}$. Therefore, this case is 
\emph{$\mathcal{I}$-non-degenerate}.

\end{enumerate}

\subsubsection{Summary Weyl type $D$:}
A Weyl type $D$  spacetime is either $\mathcal{I}$-non-degenerate or Kundt. Moreover, for a Weyl type $D$  spacetime, if $\nabla C$ is of type $I$ or more general, then it is $\mathcal{I}$-non-degenerate.

\subsection{Weyl (Petrov) type $II$}

The Weyl type $II$ tensor can be decomposed as
\[ {\sf C}={\sf N}+\sum_{A}\lambda_A\bot_A.\]
By using the annihilator operators and the projection operators we can, up to 
scaling, isolate each term in this decomposition. Each term can thus be considered a curvature operator in its own right.

In particular, by considering only the curvature tensor
$\sum_{A}\lambda_A\bot_A$, this tensor is identical to a Weyl type
$D$ tensor. We can therefore use these results. In addition to
these results we do have an additional boost weight -2 operator
${\sf N}$. This breaks the discrete symmetry present in the Weyl
type $D$ tensor and therefore restricts the choice even more.
However, with minor modifications we obtain: \emph{a Weyl type $II$
spacetime is either $\mathcal{I}$-non-degenerate or Kundt.}

We also note that for a Weyl type $II$ spacetime, the Weyl invariant $I=3\Psi_2^2$ as for type $D$. Therefore, using a similar argument, \emph{a Weyl type $II$  spacetime, if $\nabla C$ is of type $I$ or more general, then it is $\mathcal{I}$-non-degenerate. }

\subsection{Weyl (Petrov) type $III$}
For the Weyl type $III$ case we get no non-trivial curvature
operators from the Weyl tensor itself. The first non-trivial projection operators appears at first covariant derivative; however, in order to delineate this case completely we need to consider second covariant derivatives.

We note that for the type $III$ Weyl operator, ${\sf C}^2\neq 0$
and is of type $N$. The proof for this case is therefore contained
in the Weyl type $N$ case considered below.

\subsection{Weyl (Petrov) type $N$}

Consider first the tensor
\[ T^{\mu}_{~\nu}=\nabla_{\gamma}C^{\alpha\beta\gamma\mu}\nabla^{\delta}C_{\alpha\beta\delta\nu},\]
whose boost weight 0 components are of the form (it has no positive boost weight components)
\[ (T_{\mu\nu})_0\propto \bar{\kappa}^2\bar{\Psi}_4^2\bar{m}_{\mu}\bar{m}_\nu+{\kappa}^2{\Psi}_4^2{m}_{\mu}{m}_\nu.\]
Therefore, if $\kappa\neq 0$, we can construct curvature operators
of type $\{ (1,1)11\}$.
%
%
%
%
The curvature operator $T_{\mu\nu}$  gives rise to a ``Pleba\'{n}ski''
tensor of type $D$. Therefore, by considering second covariant
derivatives, it follows from the Weyl type $D$ analysis  that if
$\kappa\neq 0$, the spacetime is
\emph{$\mathcal{I}$-non-degenerate}.

Henceforth, assume that $\kappa=0$ (and therefore we have no
projectors from first derivatives). Consider $\Box
C_{\mu\nu\alpha\beta}$, which has the same symmetries as the Weyl
tensor itself. This tensor has no positive boost weight
components. Considering the boost weight 0 components, we note
that $\Box C_{\mu\nu\alpha\beta}$ is of type $II$ if and only if
$\rho\sigma\neq 0$. Therefore, if $\rho\sigma\neq 0$ we can use
the Weyl type $II$ analysis, and calculate $\nabla \Box C$; hence,
this spacetime is \emph{$\mathcal{I}$-non-degenerate}.

Therefore, consider the case where either $\sigma$ or $\rho$ are
zero. Define
\[ W^{\mu\nu}_{\phantom{\mu\nu}\epsilon\eta}\equiv C^{\alpha\beta\gamma\delta;(\mu\nu)}C_{\alpha\beta\gamma\delta;(\epsilon\eta)}.\]
To get a projection operator we note that the boost weight 0 components of $W^{\mu\alpha\beta\gamma}W_{\nu\alpha\beta\gamma}$ (it has no positive boost weight components) is of the form
\[ \left(W^{\mu\alpha\beta\gamma}W_{\nu\alpha\beta\gamma}\right)_0\propto |\Psi_4|^4|\rho \text{ or }\sigma|^8(\bar{m}^\mu m_\nu+m^{\mu} \bar{m}_{\nu})\]
Therefore, if either $\rho$ or $\sigma$ are non-zero, we can use
this operator and we get (at least) two curvature projectors
$\bot_1$ and $\bot_2$ of type $\{(1,1)(11)\}$. This means that we
can construct a Weyl-like tensor of type $D$. Hence, we can use
the type $D$ results. Therefore, by considering third derivatives
of the curvature tensors, if $\sigma$ or $\rho$ is non-zero, then
the spacetime is \emph{$\mathcal{I}$-non-degenerate}.

The remaining case, for which $\kappa=\rho=\sigma=0$, is a \emph{Kundt} spacetime.
\subsubsection{Summary Weyl type $III$ or $N$:} Therefore, we can conclude that a Weyl type $III$ or $N$ spacetime  is either $\mathcal{I}$-non-degenerate or Kundt.

\subsection{Algebraically special Ricci type}
Using the trace-free Ricci tensor, we can construct the Pleba\'{n}ski tensor, which is a Weyl-like tensor. The corresponding algebraic classification of the Pleba\'{n}ski tensor is called the Pleba\'{n}ski-Petrov (PP) classification. For the various algebraically special PP types we have the following Segre types:
\begin{itemize}
\item{} PP type $II$: $\{ 211\}$,
\item{} PP type $D$: $\{(1,1)(11)\}$, $\{(1,1)11\}$, $\{2(11)\}$, $\{z\bar{z}(11)\}$, $\{1,1(11)\}$
\item{} PP type $III$: $\{31\}$
\item{} PP type $N$: $\{(31)\}$, $\{(21)1\}$
\item{} PP type $O$: $\{(1,11)1\}$, $\{1,(111)\}$, $\{(211)\}$, $\{(1,111)\}$.
\end{itemize}
Now, since the Pleba\'{n}ski tensor is a tensor with the same
symmetries as  the Weyl tensor,  we can use the previous results
for the PP types $II$, $D$, $III$ and $N$. There is consequently
only the Weyl (Petrov) type $O$ and PP type $O$ case left to
consider.

\subsection{Weyl type $O$, PP-type $O$}

Let us consider the remaining Segre types assuming Weyl (Petrov) type $O$ and PP type $O$.
\subsubsection{Segre type $\{ (1,11)1\}$} Using the Bianchi identities we get
several differential constraints on the spin coefficients. For
this Segre type we have that $\Phi_{11}\neq 0$, so the Bianchi
identities immediately imply
$\kappa+\bar{\kappa}=\nu+\bar{\nu}=0$. In addition, we get the
following restrictions:
\[ \rho+\bar{\sigma}=s ~~(\text{real}), \quad \mu+\bar{\lambda}=m ~~(\text{real}).\] Furthermore,
after some manipulation of the remaining Bianchi identities, we
get
\[ \epsilon-\bar{\epsilon}=\gamma-\bar{\gamma}=\tau+\bar\tau+\pi+\bar\pi=0, \]
and
\[ DR=24s\Phi_{11}, \quad\Delta R=-24m\Phi_{11}.\]
We now split the analysis into 3 cases, according to whether
$\nabla_{\mu}R$ is timelike, spacelike or null.

If $\nabla_{\mu}R$ is timelike, we immediately have that this
spacetime is $\mathcal{I}$-non-degenerate since we can use
$(\bot_1)^{\mu}_{~\nu}=(\nabla^{\mu}R)(\nabla_{\nu}R)$ as a
timelike operator, and hence we obtain operators of type
$\{1,1(11)\}$.

If $\nabla_{\mu}R$ is spacelike, we can always use the remaining
freedom to choose $DR=\Delta R=0$. This implies that $m=s=0$.
Furthermore, $\delta R-\bar{\delta}R\neq 0$ which means we have an
additional spacelike projection operator. Therefore, we have a
set $\{(1,1)11\}$, which can be used to give a ``PP-type'' D
tensor. Hence, using second covariant derivatives, we find that
this is either Kundt or $\mathcal{I}$-non-degenerate.

Lastly, $\nabla_{\mu} R$ is null.  If
$\nabla_{\mu} R$ is zero, from the Bianchi identities we find that
this is a symmetric space, and hence, is actually locally
homogeneous (and Kundt). If $\nabla_{\mu}
R$ is null, we consider
$\delta\Phi_{11}-\bar{\delta}\Phi_{11}$. If this is non-zero we
get an additional projection operator and thus a set
$\{(1,1)11\}$. This would therefore give a ``PP-type'' D and
hence, by considering second derivatives, this is either Kundt or
$\mathcal{I}$-non-degenerate. Therefore, let us assume
$\delta\Phi_{11}-\bar{\delta}\Phi_{11}=0$. The Bianchi identities
now imply that $\alpha-\bar\beta=0$. Using the symmetric operator
$R^{\alpha\beta;\nu}R_{\alpha\beta;\mu}$ we get that this is
either of types $\{2(11)\}$, $\{(1,1)(11)\}$ or $\{(211)\}$. The
first two of these give a type $D$ ``Pleba\'{n}ski'' tensor which
means, by considering second derivatives, they are either
$\mathcal{I}$-non-degenerate or Kundt. For the last case,
$\{(211)\}$, we can combine with the Ricci operator to break the
symmetry down to type $\{(21)1\}$ which gives a ``Pleba\'{n}ski'' of
type $N$. Therefore, by considering third and fourth  derivatives, we get that this is
$\mathcal{I}$-non-degenerate or Kundt.
{\footnote{Note that this is the only case in which we have needed to utilize fourth
derivatives; it is possible, by explicitly calculating the components of
$R_{\mu\nu;\alpha\beta\gamma}$, that we need only consider third derivatives, but
we have not done this here.  }}

\subsubsection{Segre type $\{1,(111)\}$} This is actually Ricci type $I$ and is therefore $\mathcal{I}$-non-degenerate.
\subsubsection{Segre type $\{ (211)\}$} Choosing a frame where $\Phi_{22}$ is a constant
we get, after using the Bianchi identities (and some
manipulation),  $\kappa=\sigma=0$. We then calculate the second
derivatives and compute the operator
$R^{\alpha\beta;(\gamma\mu)}R_{\alpha\beta;(\gamma\nu)}$, which
gives operators of type $\{(1,1)(11)\}$, assuming $\rho\neq 0$.
This  again gives ``PP-type'' D tensor and hence, by calculating
third derivatives, this is $\mathcal{I}$-non-degenerate or Kundt.
\subsubsection{Segre type $\{(1,111)\}$} This is the maximally symmetric case and is therefore  Kundt.\\

In addition to Weyl-type $I$, Ricci-type $I$ or Riemann-type $I/G$, we have shown that there are $\mathcal{I}$-non-degenerate metrics with algebraically special curvature types and further conditions on the spin-coefficients.  These are summarized in Tables (\ref{table:1}) and (\ref{table:2}).

\begin{table}[h]
\begin{center}
\small{
\begin{tabular}{c|c}
\hline
 & \\
{\bf P-type} & {\bf Conditions}  \\
 & \\
\hline
 & \\
I & --- \\ & \\ \hline & \\
D or II & $\begin{array}{c}
                \kappa \neq 0 \\
                \kappa =0,\, (|\sigma|^2+|\rho|^2)(|\lambda|^2+|\mu|^2) \neq 0 \\
                \kappa=\lambda=\mu=0,\, |\sigma|^2+|\rho|^2 \neq 0
            \end{array}$ \\
 & \\ \hline & \\
N or III & $ \begin{array}{c}
                \kappa \neq 0 \\
                \kappa=0\, ;\, \sigma=0 \ \textnormal{or}\ \rho=0\ \textnormal{(not both)}
             \end{array}$ \\
 & \\ \hline
\end{tabular}
}
\caption{In each Petrov-type we list conditions on the spin-coefficients yielding distinct subcases for which the metric is $\mathcal{I}$-non-degenerate.}\label{table:1}
\end{center}
\end{table}

\begin{table}[h]
\begin{center}
\small{
\begin{tabular}{c|c}
\hline
 & \\
{\bf Segre type of $R_{\mu\nu}$} & {\bf Conditions}  \\
 & \\
\hline
 & \\
$\{1,(111)\}$ & --- \\ & \\ \hline & \\
$\{(211)\}$ & $\Phi_{22}$ \textnormal{const.},\, $\kappa=\sigma=0,\, \rho \neq 0$\ \textnormal{(3rd deriv.)} \\
& \\ \hline
& \\
$\{(1,11)1\}$ & $\begin{array}{c}
                \nabla_{\mu}R\ \textnormal{timelike} \\
                \\ \hline \\
                \nabla_{\mu}R\ \textnormal{spacelike}:\, DR=\Delta R=\kappa+\bar\kappa=\rho+\bar\sigma=\mu+\bar\lambda=0\ \\
                \textnormal{(2nd deriv. ; at least one of }\kappa,\, \sigma\ \textnormal{or}\ \rho\ \textnormal{is nonzero)} \\ \\ \hline \\
                \nabla_{\mu}R\ \textnormal{null}: \begin{array}{l}
                                                    \nabla_{\mu}R=0\ \textnormal{(symmetric space)} \\ \\ \hline \\
                                                    \nabla_{\mu}R\neq 0: \begin{array}{l}
                                                                            \delta\Phi_{11}-\bar\delta\Phi_{11} \neq 0\ \textnormal{(2nd deriv.)} \\
                                                                            \delta\Phi_{11}-\bar\delta\Phi_{11}= \alpha-\bar\beta=0\ \\ \textnormal{(2nd or 3rd and 4th deriv.)}
                                                                        \end{array}

                                                    \end{array}
            \end{array}$ \\
 & \\ \hline
\end{tabular}
} \caption{Within P-type $O$ and PP-type $O$ we list the Segre types
that contain $\mathcal{I}$-non-degenerate metrics. The $n$th
derivative conditions indicate that higher order constraints exist
on the spin-coefficients arising from $n$th order curvature
operators. These higher-order constraints provide sufficient
conditions for the metric to be $\mathcal{I}$-non-degenerate.  In
all cases, at least one of $\kappa$, $\sigma$ or $\rho$ is
nonzero.}\label{table:2}
\end{center}
\end{table}

\section{Curvature invariants}
We have addressed the question of what is the class of Lorentzian
manifolds that can be completely characterized by the scalar
polynomial invariants constructed from the Riemann tensor and its
covariant derivatives. Let us now consider the 'inverse' question:
given a set of scalar polynomial invariants, what can we say about
the underlying spacetime? In practice, it is somewhat tedious and
a lengthy ordeal to determine the spacetime from the set of
invariants. However, in most circumstances we only need some
partial results or we are dealing with special cases. Let us
discuss how to determine, from the invariants, whether the
spacetime is $\mathcal{I}$-non-degenerate.

We remind the reader that the zeroth order Weyl invariants are
$I$ and $J$, and if all Weyl invariants up to order $k$ vanish, we
will denote this by VSI$^W_{k}$.

\begin{prop}
If $27J^2\neq I^3$, then the spacetime is $\mathcal{I}$-non-degenerate.
\end{prop}
This follows easily from the fact that
if $27J^2\neq I^3$ then the
spacetime is of Weyl (Petrov) type I.

If $27J^2= I^3\neq 0$ (Weyl type $II$ or D), we need to go to higher
order invariants in order to check whether it is
$\mathcal{I}$-non-degenerate or not. Ideally, we would like to
have a set of syzygies which gives the appropriate condition for
this to be the case. Such a complete set is not known. However, we
have found two such syzygies which gives a \emph{sufficient}
condition for $\mathcal{I}$-non-degeneracy.  A number of
invariants of $\nabla C$ were constructed with degrees ranging
from 2 to 4 (see Appendix \ref{nabcinvar} for details).  Imposing the minimal number of conditions required
for the normal form of a $\nabla C$-type $II$ (boost weight +3,+2,+1
vanish) or D (only boost weight 0 is nonzero) results in a degree
8 syzygy, $S_1=0$, and a degree 16 syzygy, $S_2=0$, amongst our
invariants.  Therefore if $S_1\neq 0$ or $S_2\neq 0$ then $\nabla
C$ is not of type $II$ or $D$.  Next, we showed that using the normal
form of a $\nabla C$-type $G$ (all components nonzero) or $H$ (boost
weight +3 vanish) or $I$ (boost weight +3, +2 vanish) then
$S_{1}\neq 0$ and $S_{2}\neq 0$.  It is important to note that
this implication refers only to the general types of $G$, $H$ and $I$
and there is no consideration of a secondary alignment type or any
further algebraic specialization within these types. Indeed, it is
possible that there is an algebraically special subcase, for
example of a $\nabla C$-type I, that results in $S_{1}=S_{2}=0$. A
stronger statement relating invariants of $\nabla C$ to its
algebraic type may be achieved by considering a different basis of
invariants and a finer algebraic classification of $\nabla C$.  Initially, one would attempt to construct a set of pure $\nabla C$ invariants that was complete within each $\nabla C$ algebraic type $G$, $H$, $I$ and $II$, including special subcases.  We have excluded type $D$ since such a set of invariants is equivalent to type $II$, and also types $III$, $N$ or $O$ since these invariants vanish.  By completeness of the set, an algebraic specialization would result in a dependence amongst invariants and hence syzygies arise characterizing the algebraically special type.  We now have the following invariant characterizations of $\mathcal{I}$-non-degeneracy.

\begin{prop}
If $27J^2= I^3$, but $S_1\neq 0$ or $S_2\neq 0$, then the spacetime is $\mathcal{I}$-non-degenerate.
\end{prop}
The remaining cases are when both $I$ and $J$ are zero, and hence, the spacetime is VSI$^W_{0}$:
\begin{prop}
Assume a spacetime is VSI$^W_{0}$. Then:
\begin{enumerate}
\item{} If it is \emph{not} VSI$^W_{1}$, it is $\mathcal{I}$-non-degenerate.
 \item{} If it is VSI$^W_{1}$, but \emph{not} VSI$^W_{2}$, it is $\mathcal{I}$-non-degenerate.
\end{enumerate}
\end{prop}

To prove the final result below, we shall assume for simplicity that the
spacetime is Einstein, so that $R_{\mu\nu}=\lambda g_{\mu\nu}$. We
therefore only have to consider the Ricci scalar ($=4\lambda$) and
the Weyl invariants. If this is not the case, then we would need
to include the Ricci and mixed invariants. This can be done in a
straight-forward manner.  A summary of these results is given in Figure
\ref{flowchart}.
\begin{prop}
Assume a spacetime is Einstein. Then:
\newline
\indent 3. If it is VSI$^W_{2}$, then it is Kundt.
\end{prop}

\begin{figure}[t]
\scalebox{1.6}{\includegraphics[height=80mm,width=60mm]{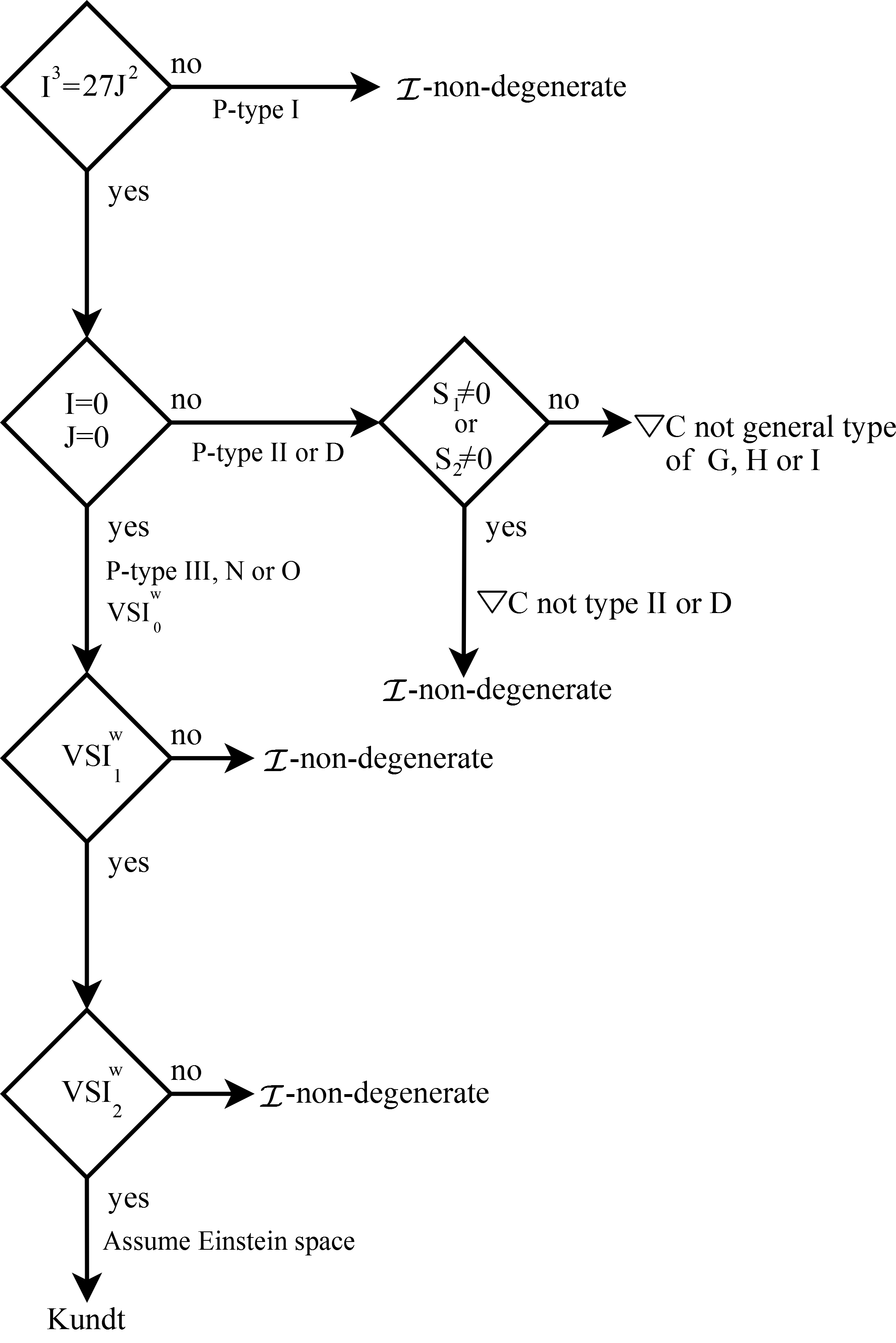}}
\caption{Using invariants in terms of the Weyl tensor and its
covariant derivatives to determine whether the spacetime is
$\mathcal{I}$-non-degenerate.}\label{flowchart}
\end{figure}


From the above results we have conditions on the scalar invariants
(in terms of the Weyl tensor and its covariant derivatives) to
determine whether the spacetime is $\mathcal{I}$-non-degenerate.
Consequently, we have a number of conditions in terms of scalar
invariants that can be used to determine when a spacetime is not
$\mathcal{I}$-non-degenerate and hence an aligned algebraically
special type-$II$  (or {\it degenerate}) Kundt spacetime.

Let us further consider to what extent the class of degenerate
Kundt spacetimes can be characterized by their scalar curvature
invariants. Clearly such spacetimes are algebraically special and
of type $II$ (or more special) and hence  $27J^2= I^3$. If $I=J=0$,
then if the spacetime is of Weyl type $N$, then
$\mathcal{I}_1=\mathcal{I}_2=0$ if and only if
$\kappa=\rho=\sigma=0$ from the results in \cite{4DVSI} (the
definitions of the invariants $\mathcal{I}_1$ and $ \mathcal{I}_2$
are given therein). Similar results follow for Weyl type $III$
spacetimes (in terms of the invariants $\tilde{\mathcal{I}}_1$ and
$\tilde{\mathcal{I}}_2$) and in the conformally flat (but
non-vacuum) case (in terms of similar invariants $\mathcal{I}_1$
and $\mathcal{I}_2$ constructed from the Ricci tensor)
\cite{4DVSI}. If $27J^2= I^3\neq 0$ (Weyl types $II$ and $D$):
essentially if $\kappa=\rho=\sigma \neq 0$, we can construct
positive boost weight terms in the derivatives of the curvature
and determine an appropriate set of scalar curvature invariants.
For example, consider the positive boost weight terms of the first
covariant derivative of the Riemann tensor, $\nabla(Riem)$. If the
spacetime is $\mathcal{I}$-non-degenerate, then each component of
$\nabla(Riem)$ is related to a scalar curvature invariant. In this
case, in principle  we can solve (for the positive boost weight
components of $\nabla(Riem)$) to uniquely determine
$\kappa,\rho,\sigma$ in terms of scalar invariants, and we can
therefore find necessary conditions for the spacetime to be
degenerate Kundt (there are two cases to consider, corresponding
to whether $\Psi_2 + \frac{2}{3} \Phi_{11}$ is zero or non-zero).
We note that even if the invariants exist in principle, it may not
be possible to construct them in practice.

\section{Weakly and Strongly $\mathcal{I}$-non-degenerate}
Until now we have only considered $\mathcal{I}$-non-degeneracy in
terms of a local deformation of the metric. It is also of interest
to know whether a $\mathcal{I}$-non-degenerate metric is unique
under a discrete transformation. We shall call a spacetime such
that the set of invariants  \emph{uniquely} specifies the metric
\emph{strongly $\mathcal{I}$-non-degenerate}. Similarly, we shall call a
spacetime such that the set of invariants only defines a
unique metric up to discrete transformations \emph{weakly
$\mathcal{I}$-non-degenerate}.

Let us revisit the examples given by eqs. (\ref{weakly1}) and
(\ref{weakly2}) in the Introduction. These two examples are both of
Weyl type $O$, but they are of Segre type $\{1,(111)\}$ and
$\{(1,11)1\}$. Hence, the eigenvalues of the Ricci operator is the
same but we cannot, from the invariants alone, determine which
eigenvalue is associated with the timelike direction and which is
associated with the spacelike direction. This is linked to the
fact that the map where we swap time with a space direction is not
a Lorentz transformation. Note that permuting  any two
axes in the Riemannian-signature case is an $O(n)$ transformation,
while permuting time and space in the Lorentzian case is
\emph{not} an $O(1,n-1)$ transformation. Therefore, there is no
distinction between weakly and strongly
$\mathcal{I}$-non-degenerate in the Riemannian case.

In most cases we do actually have a frame in which  we know
which direction is time. However, if we are only handed a set of
invariants we would not have such  a frame and, a priori, we would not
know which eigenvalue is associated with time. We also note that
the ambiguity in choosing a projection operator in certain cases
is linked to the same problem; we do not necessarily  know which
eigenvalue is associated with time.

Therefore, the question of which $\mathcal{I}$-non-degenerate
metrics are strongly $\mathcal{I}$-non-degenerate is linked to the
question of when the time direction can be uniquely specified from
the set of invariants.

Consider an invariant $I$. Then we can consider the gradient,
$v_{\mu}\equiv\nabla_{\mu}I$, which is a curvature ``vector''.
Assume that the metric is $\mathcal{I}$-non-degenerate, in which
case we always have a timelike projection operator, $\bot_1$.
Therefore, we can consider $\bot_1{\bf v}$. Now, if
$\bot_1{\bf v}\neq 0$ then clearly it is timelike and the invariant
$(\bot_1{\bf v})^\mu (\bot_1{\bf v})_\mu <0$. Therefore, we could
uniquely specify time, because $\bot_1{\bf v}$ would give us the
time direction. So if there exists an invariant $I$ for which
$({\bot_1})^{\mu}_{~\nu}\nabla^\nu I$ is timelike (and non-zero),
this spacetime is strongly $\mathcal{I}$-non-degenerate.

A similar conclusion is reached if we have three spacelike
projection operators and all of these have similar non-zero
gradients.  To be more precise:
\begin{prop}
Consider a (weakly) $\mathcal{I}$-non-degenerate spacetime. Then, if either:
\begin{enumerate}
\item{} there exists an invariant $I=v_{\mu}v^{\mu}$, where
$v_{\mu}$ is a curvature 1-tensor, such that $I<0$; or, \item{}
there exist curvature 1-tensors $v_{\mu}$, $u_{\mu}$ and $w_{\mu}$
such that the invariants $I_1=v_{\mu}v^{\mu}>0$,
$I_2=u_{\mu}u^{\mu}>0$, $I_3=w_{\mu}w^{\mu}>0$, and
\[ I_4=v_{\alpha}u_{\beta}w_{\gamma}g^{[\alpha}_{~\mu}g^{\beta}_{~\nu}g^{\gamma]}_{~\lambda}v^{\mu}u^{\nu}w^{\lambda}\neq 0;
\]
\end{enumerate}
then the spacetime is \emph{strongly} $\mathcal{I}$-non-degenerate.
\end{prop}
\begin{proof}
In case (1) we can construct a timelike projection operator, and the
result follows. In case (2) there exist three spacelike projection
operators, and the condition that $I_4\neq 0$ ensures that these are linearly
independent. Hence, the timelike vector is orthogonal to these
three and the result follows.
\end{proof}
Therefore, the only spacetimes that are weakly
$\mathcal{I}$-non-degenerate but not strongly
$\mathcal{I}$-non-degenerate must have a timelike and a
spacelike derivative which annihilate all invariants. If the
spacetime is weakly $\mathcal{I}$-non-degenerate, but not strongly
$\mathcal{I}$-non-degenerate, there must consequently exist a
timelike vector, ${\mbold\xi}_1$, and a spacelike vector,
${\mbold\xi}_2$, for which
\[{\mbold\xi}_{1}(I)={\mbold\xi}_2(I)=0,\]
for all scalar invariants $I$. If
$[{\mbold\xi}_1,{\mbold\xi}_2]={\mbold\xi}_3$, it also follows that
${\mbold\xi}_3(I)=0$. Therefore, there will be a set of
vectors,  $\{ {\mbold\xi}_i\}$,  closed under commutation
(consequently, the Jacobi identity will also be satisfied), which
annihilates all curvature invariants. This has several
consequences. First, this set will span a timelike (sub)manifold
of dimension 2, 3 or 4. We can therefore locally introduce normal
coordinates, so that the invariants only depend on the normal
coordinates; i.e., $I=I(x,y)$ (dim 2), $I=I(x)$ (dim 3) or
$I=\text{constant}$  (dim 4, and the spacetime is a CSI
spacetime). Second, by the assumption that this spacetime is
weakly $\mathcal{I}$-non-degenerate, and the fact that these
invariants only depend on the coordinates $(x,y)$, there
exists an orthonormal frame such that \emph{all components of the
curvature tensors only depend on the normal coordinates $(x,y)$}
\cite{KN,kramer}.

This indicates that these vectors that annihilate all invariants
have a special geometric meaning. First, let us consider an
arbitrary curvature tensor of rank $(n,n)$,
${\mathcal{R}^{\alpha_1\dots\alpha_n}}_{\beta_1\dots\beta_n}$,  being a sum,
tensor products and contractions of the Riemann tensor and its
covariant derivatives. Since this tensor has as many covariant as
contravariant indices, we can interpret this as a curvature
operator, ${\sf R}\equiv{\mathcal{R}^{\alpha_1\dots\alpha_n}}_{\beta_1\dots\beta_n}{\bf e}_{\alpha_1}\otimes\cdots\otimes{\bf e}_{\alpha_n}\otimes{\mbold\omega}^{\beta_1}\otimes\cdots\otimes{\mbold\omega}^{\beta_n}$,  mapping rank $n$ contravariant tensors into rank $n$
contravariant tensors. Let us denote ${\bf T}^{End}$ as the tensor
algebra of \emph{all} such curvature operators. It is clear that
all polynomial curvature invariants can be considered as complete
contractions of operators in  ${\bf T}^{End}$.
\begin{thm}\label{Liexi}
Consider a spacetime which is (weakly) $\mathcal{I}$-non-degenerate,
and a vector field ${\mbold\xi}$. Then the following conditions are
locally equivalent:
\begin{enumerate}
\item{} ${\mbold\xi}(I)=0$ for all curvature invariants $I$.
\item{} The Lie derivative of any curvature operator ${\sf R}\in{\bf T}^{End}$ with respect to ${\mbold\xi}$, vanishes; i.e.,
\[\pounds_{{\mbold\xi}}{\sf R}=0.\]
\end{enumerate}
\end{thm}
\begin{proof}
\emph{(1) $\Rightarrow$ (2):} Assume that ${\mbold\xi}(I)=0$ for all curvature invariants $I$. Consider the 1-parameter group of diffeomorphisms, $\phi_t$, associated with the vector field ${\mbold\xi}$. Then
\[ {\mbold\xi}(I)=\pounds_{\mbold\xi}(I)=\lim_{t\rightarrow 0}\frac 1t[I-\phi_t^*(I)]=0.\]
Assuming the conditions hold over a neighborhood $U$, this can be
integrated and we get, at a point $p\in U$, $I(p)=I(\phi_{-t}(p))$.
Hence, along the integral curves the value of the invariants do not
change. Consider now the Lie derivative of an arbitrary curvature
operator ${\sf R}={\mathcal{R}^{\alpha_1\dots\alpha_n}}_{\beta_1\dots\beta_n}{\bf e}_{\alpha_1}\otimes\cdots\otimes{\bf e}_{\alpha_n}\otimes{\mbold\omega}^{\beta_1}\otimes\cdots\otimes{\mbold\omega}^{\beta_n}$
(e.g., see \cite{KN}):
\[ \pounds_{{\mbold\xi}}{\sf R}=\lim_{t\rightarrow 0}\frac 1t\left[{\sf R}-\hat{{\sf R}}_t\right], \]
where $\hat{{\sf R}}_t$ is the $\phi_t$-transformed tensor defined by:
\[ \hat{{\sf R}}_t=\phi_t({\sf R})=\phi_t^*({\mathcal{R}^{\alpha_1\dots\alpha_n}}_{\beta_1\dots\beta_n})(\phi_{t*}{\bf e}_{\alpha_1})\otimes\cdots\otimes(\phi_{t*}{\bf e}_{\alpha_n})\otimes(\phi_t^*{\mbold\omega}^{\beta_1})\otimes\cdots\otimes(\phi_t^*{\mbold\omega}^{\beta_n}).
\]
The action of $\phi_t$ preserves the form and symmetries of a tensor.
Thus the transformed tensor $\hat{{\sf R}}_t$ will be a curvature tensor
of the same kind as ${\sf R}$. The curvature invariants at $p$ will be
$I(p)$ for ${\sf R}$ and $I(\phi_{-t}(p))$ for $\hat{{\sf R}}_t$.
From the above, these invariants are the same and, from the assumption
of $\mathcal{I}$-non-degeneracy, the invariants characterise
the spacetime, which means that there exists a frame such that the components of the curvature tensors do not change along $\phi_{-t}(p)$. This frame essentially is the eigenvalue frame of the curvature tensors. In particular, the projection operators define this frame.

If ${\sf v}$ is an eigenvector of ${\sf R}$, then
\[{\sf R}{\sf v}-\lambda{\sf v}=0, \quad \Rightarrow \quad  \hat{\sf R}_t\hat{\sf v}_t-\lambda\hat{\sf v}_t=0,\]
where hatted quantities are transformed under $\phi_t$. Eigenvectors
are therefore transformed onto eigenvectors of $\hat{{\sf R}}_t$.
Using the fact that there exists a frame so that
$\phi_t^*({\mathcal{R}^{\alpha_1\dots\alpha_n}}_{\beta_1\dots\beta_n})={\mathcal{R}^{\alpha_1\dots\alpha_n}}_{\beta_1\dots\beta_n}$,
means that the components remain the same in this frame.
For a symmetric operator the eigenvectors are orthogonal and we
can introduce a basis of orthonormal eigenvectors $\{ {\bf e}_I\}$
with duals $\{ {\mbold\omega}^I\}$. Consider now a symmetric
projection operator, ${\bot}$, written in the eigenvector basis:
\[ {\bot}=\delta^A_{~B}{\bf e}_A\otimes{\mbold\omega}^B, \quad \hat{\bot}_t=\delta^A_{~B}\hat{\bf e}_A\otimes\hat{\mbold\omega}^B,\]
where the indices run over a subset of eigenvectors with the same
eigenvalue, and the hatted basis is the transformed basis. From
the above discussion we see that the eigenspaces are
$\phi_t$-invariant, and hence there is a transformation matrix
$M^A_{~B}$ such that $\hat{\bf e}_{A}=M^{\tilde{A}}_{~A}{\bf
e}_{\tilde{A}}$, and
$\hat{\mbold\omega}^B=(M^{-1})^{B}_{~\tilde{B}}{\mbold\omega}^{\tilde{B}}$.
Consequently,
\[
{\bot}=\hat{\bot}_t,
\]
and the curvature projection operators are $\phi_t$-invariant.
Therefore, since all ${\sf R}\in {\bf T}^{End}$ can be expanded in
terms of these projection operators and the curvature invariants
(since it is $\mathcal{I}$-non-degenerate), we have that ${\sf
R}=\hat{\sf R}_t$ and (2) follows.

\emph{(2) $\Rightarrow$ (1):} This follows trivially from the
observation that $ {\mbold\xi}(I)=\pounds_{\mbold\xi}(I)$ and the
properties of the Lie derivative.
\end{proof}
\begin{cor}
If there exists a non-zero vector field,  ${\mbold\xi}$, fulfilling
\[\pounds_{{\mbold\xi}}{\sf R}=0,\]
for all curvature operators ${\sf R}\in {\bf T}^{End}$,
 then the spacetime possesses a Killing vector field, $\hat{\mbold\xi}$.
\end{cor}
\begin{proof} This follows from the equivalence principle \cite{kramer}.
The Cartan scalars are related to the components of the Riemann tensor and
its derivatives, and along the integral curves of ${\mbold\xi}$ we can use
$\phi_{t}$ at any given point $p$.  We want to compare the tensors at $p$ and $q\equiv\phi_{t}(p)$. Consider an arbitrary even-ranked curvature tensor ${\bf R}$. By raising or
lowering indices appropriately, we get an operator ${\sf R}$. Since the Lie derivative of ${\sf R}$ along ${\mbold\xi}$ vanishes, there is a frame such that ${\sf R}_q$ and ${\sf R}_p$ has identical components. Therefore, by raising and lowering the indices appropriately, the components of ${\bf R}_q$ and ${\bf R}_p$ are also the same. The Cartan invariants of  ${\bf R}_q$ and ${\bf R}_p$ are therefore the same. For a curvature tensor, ${\bf R}$,
of odd rank we consider ${\bf R}\otimes{\bf R}$, which is of even
rank, and use the
fact that $\phi_{t}$ is continuous in $t$.  Therefore, there exists a frame such that all the components of any curvature tensor are identical at $p$ and $q$. The equivalence principle now implies that $\phi_{t}$, for
any given $t$, is an isometry; hence, there must exist a Killing vector
field $\hat{\mbold\xi}$ which generates an isometry $\hat{\phi}_{\hat{t}}$
such that $\hat{\phi}_{\hat{t}}(p)=\phi_{t}(p)$.
\end{proof}

Note that in most cases $\hat{\mbold\xi}$ and ${\mbold\xi}$ are the same.
However, in some very special cases with additional symmetries they need
not be (although locally they are of the same causality; e.g., they are both
timelike or both spacelike).  For example, for flat
space the curvature vanishes identically; hence,
$\pounds_{{\mbold\xi}}{\bf R}=0$ for all ${\mbold\xi}$ and any curvature
tensor ${\bf R}$, although not all ${\mbold\xi}$ are Killing vectors.
However, in these special cases there will always exist
at least two Killing vectors.

Therefore, to conclude:
\begin{cor} If a spacetime is \emph{weakly}
$\mathcal{I}$-non-degenerate but not \emph{strongly}
$\mathcal{I}$-non-degenerate, then it possesses locally (at least) one
timelike Killing vector and one spacelike Killing vector.
\end{cor}

\section{Conclusions}

In this paper we have addressed the question of what is the class
of Lorentzian manifolds that can be completely characterized by
the scalar polynomial invariants constructed from the Riemann
tensor and its covariant derivatives.
In the Riemannian case the manifold is always locally
characterized by the scalar polynomial invariants and, therefore,
all of the Cartan invariants are related to the scalar
curvature invariants \cite{kramer}. We have generalized these results
to the Lorentzian case.

We have introduced the
important notion of $\mathcal{I}$-non-degenerate spacetime
metrics. In order to prove the main theorems, which is done on a
case-by-case (depending on the algebraic type) using a
boost weight decomposition, we have introduced an appropriate set
of curvature operators and curvature projectors.
In the (algebraically) general case we have shown that if a 4D
spacetime is either Ricci type $I$, Weyl type $I$ or Riemann type $I/G$,
then it is $\mathcal{I}$-non-degenerate, which implies that the
spacetime metric is determined by its curvature invariants (at
least locally, in the sense explained above).

For the algebraically special cases the Riemann tensor itself does
not give enough information to provide us with all the required
projection operators, and it is also necessary to consider the
covariant derivatives. In terms of the boost weight decomposition,
for an algebraically special metric (which has a Riemann tensor
with zero positive boost weight components) which is not Kundt, by
taking covariant derivatives of the Riemann tensor positive boost
weight components are acquired and a set of higher derivative
projection operators are obtained. Consequently, we found that if
the spacetime metric is algebraically special, but $\nabla R$,
$\nabla^{(2)} R$, $\nabla^{(3)} R$ or $\nabla^{(4)} R$ is of type I or more
general, the metric is $\mathcal{I}$-non-degenerate.

The remaining metrics which \emph{do not} acquire a positive
boost weight component when taking covariant derivatives have a
very special curvature  structure. Indeed, in our main theorem we
proved that a spacetime metric is either
$\mathcal{I}$-non-degenerate or the metric is a Kundt metric. This
is very striking result because it implies that a metric that is
not determined by its scalar curvature invariants must be of Kundt
form. The Kundt metrics which are not $\mathcal{I}$-non-degenerate therefore correspond to degenerate
metrics in the sense that many such metrics can have identical
scalar invariants.
This exceptional property of the
the degenerate Kundt metrics essentially follows from the fact that
they do not define a unique timelike curvature operator.

The results in the case of Petrov type I spacetimes in 4D follow
from the above theorems. Although these results were not
previously known, some partial results for 4D Weyl (Petrov) type
I spacetimes, which are consistent with the above analysis, can be
deduced from previous work. This is discussed in the next section
(also see Appendix D).

Therefore, if a spacetime is $\mathcal{I}$-non-degenerate and the
algebraic type is explicitly known (using, for example,  the
Pleba\'{n}ski notion for the Segre type in which commas are used to
distinguish between timelike and spacelike eigenvectors and their
associated eigenvalues, as is common in general relativity), the
spacetime can be completely classified in terms of its scalar
curvature invariants.

There are a number of important consequences of the results obtained. A
corollary of the main theorem applied to spacetimes with
constant curvature invariants (CSI) is a proof of the CSI-Kundt
conjecture in
4D \cite{CSI3D}. In future work we will study CSI spacetimes
in more detail \cite{CSI4}.

We then considered the inverse question: given a set of scalar
polynomial invariants, what can we say about the underlying
spacetime? In 4D we can partially characterize the Petrov type in
terms of scalar curvature invariants. In most circumstances we
only need some partial results or necessary conditions. For
example, we found that if $27J^2\neq I^3$, or if $27J^2= I^3$ but
the invariants $S_1\neq 0$ or $S_2\neq 0$, then the spacetime is
$\mathcal{I}$-non-degenerate. Some results were then presented in
the remaining cases  when both $I$ and $J$ are zero, and hence the
spacetime is VSI$^W_{0}$.

We also discussed whether a $\mathcal{I}$-non-degenerate metric is
unique under a discrete transformation. We introduced the notion
\emph{strong} and \emph{weak} non-degeneracy. We provided a
necessary criterion to determine spacetimes that are weakly
$\mathcal{I}$-non-degenerate but not strongly
$\mathcal{I}$-non-degenerate .

Having determined when a spacetime is completely characterized by its scalar
curvature invariants, it is also of interest to determine the minimal set of
such invariants needed for this classification. For example, in 4D there are
results
on determining the Riemann tensor in terms of zeroth order scalar
curvature invariants (and determining a minimal set of such
invariants) \cite{Carminati}.
It is also of interest to study when
a spacetime can be explicitly {\em constructed} from scalar curvature invariants.


This work is also of importance to the equivalence problem of
characterizing Lorentzian spacetimes (in terms of their Cartan scalars) \cite{kramer}.
Clearly, by knowing which spacetimes can be characterized
by their scalar curvature invariants alone, the computations
of the invariants (i.e., simple polynomial scalar invariants) is much more
straightforward and can be done algorithmically (i.e., the full
complexity of the equivalence method is not necessary). On the other
hand, the Cartan equivalence method also contains, at least in principle,
the conditions under which the classification is complete (although
in practice carrying out the classification for the more general
spacetimes is difficult, if not impossible). Therefore, in a sense,
the full machinery of the Cartan equivalence method is only necessary
for the classification of the degenerate Kundt spacetimes (which we shall
address in future work).

Let us briefly discuss this further in the context of
two simple examples, which also serve to illustrate the results
of the main theorem:

\begin{enumerate}

\item{}
The Schwarzschild vacuum type $D$ spacetime is an example of an
$\mathcal{I}$-non-degenerate spacetime. In the canonical
coordinate form of the  metric as given in \cite{Paiva}, the two scalar polynomial
invariants $C^2 \equiv C_{abcd}C^{abcd} = 48{m^2}{r^{-6}}$ and
$(\nabla C)^2 \equiv C_{abcd;e}C^{abcd;e} = 720(r-2m){m^2}{r^{-9}}$ are functionally
independent and can be used to solve for $r$ and $m$, and all of the
algebraically independent Cartan scalars
$\Psi_{2}$, $\nabla^2 \Psi_{20'}$, $\nabla^2 \Psi_{31'}$, and
$\nabla^2 \Psi_{42'}$  are consequently related to the polynomial
curvature invariants $C^2$ and $(\nabla C)^2$
\cite{Paiva}. In particular,
$\Psi_{2}= -{m}r^{-3}$, $\nabla^2 \Psi_{20'}= 12{m^2}r^{-6} -
6mr^{-5}$, so that
$48(\Psi_{2})^2 = C^2$   and
$120 (\Psi_{2}) (\nabla^2\Psi_{20'}) = -(\nabla C)^2$.  
We note that the second derivative Cartan scalars have the following boost weights:
$\nabla^2 \Psi_{20'}$ is +2, $\nabla^2 \Psi_{42'}$ is -2 and $\nabla^2 \Psi_{31'}$ is 0.

\item{}
A spatially homogeneous vacuum plane wave, which is a
special subcase of a Petrov type $N$ vacuum spacetime
admitting a covariantly constant null vector, belongs to the
class of vanishing scalar invariant (VSI) spacetimes \cite{4DVSI} and is consequently an example of a
degenerate Kundt spacetime. Since it is a VSI spacetime,
all scalar polynomial invariants are zero. However, distinct VSI
spacetimes give rise to a
distinct set of Cartan scalars \cite{kramer} (e.g., in flat space all
of the Cartan scalars
are zero). A spatially homogeneous vacuum plane wave
has two non-trivial Cartan scalars,
$\nabla \Psi_{00'}$ and $\nabla^2 \Psi_{00'}$.

\end{enumerate}


\section{Discussion}

We have addressed the question of what is the class of Lorentzian
manifolds that can be completely characterized by the scalar
polynomial invariants constructed from the Riemann tensor and its
covariant derivatives. In particular, we proved the result that
this is true in the case of Petrov type I spacetimes in 4D. This
result was not previously known. However, some partial results for
4D Weyl (Petrov) type I spacetimes are known, which are
consistent with the above analysis. Let us review these results.

Essentially, in the case of Petrov (Weyl) type I, there exists a
unique frame so that all components of the Riemann tensor are
related to curvature invariants.  Indeed, in general there are
four different curvature invariants (e.g., corresponding to the
complex invariants $I$ and $J$), so that all invariants (which depend
on $4$ coordinates) are functionally dependent on these four
invariants. Problems arise in degenerate cases and cases with
symmetries. It is also known that all Petrov type I spacetimes are
completely backsolvable \cite{Carminati}.

Let us consider the Petrov type I case in more detail. From
\cite{HALL, HallBook} (also see Appendix D) it follows that if a
4D spacetime is of Petrov type I it can be classified according
to its rank and it is either:

\begin{enumerate}

\item curvature class A (and the holonomy group is general and of
type $R_{15}$),

\item curvature class C (and of holonomy type $R_{10}$ or
$R_{13}$, with restricted Segre type).

\end{enumerate}

Now, suppose the components of the Riemann tensor $R^a \;\!
_{bcd}$ are given in a coordinate domain $U$ with metric $g$. In
case (1), where the curvature class is of type A, for any other
metric $g'$ with the same components $R^a \;\! _{bcd}$ it follows
that $g'_{ab} = \alpha g_{ab}$ (where $\alpha$ is a constant);
i.e., the metric is determined up to a constant conformal factor
and the connection is uniquely determined.  This implies that all
higher order covariant derivatives of the Riemann tensor are
completely determined; i.e., given $R^a \;\! _{bcd}$, all of the
components of the covariant derivatives are determined and we only
need classify the Riemann tensor itself.  (Note that all of the
scalar polynomial curvature invariants are then determined, at
least up to an overall constant factor).

We can then pass to the frame formalism and determine the frame
components of the Riemann tensor (to do this we need the metric to
determine the orthogonality of the frame vectors and hence
construct the frame; since $g$ is specified up to an overall
constant conformal factor, orthogonality is unique).  The Petrov
type I case is completely backsolvable \cite{Carminati} and hence
the frame components are completely determined by the zeroth order
scalar invariants.  Therefore, it follows that the spacetime is
completely characterized by its scalar curvature invariants in
this case.

Let us now consider case (2), where the curvature
class is $C$.  Again, let us suppose that the $R^a \;\! _{bcd}$ are given
in $U$ with metric $g$.  If $g'$ is any other metric with the same
$R^a \;\! _{bcd}$, it follows that
$$  g'_{ab} = \alpha g_{ab} + \beta k_ak_b  $$
(where $\alpha$ and $\beta$ are constants).  The equation
\beq\label{star} R^a \;\! _{bcd} k^d =  0,  \eeq has a unique
non-trivial solution for $k \in T_m M$. Note that $R^a\;\! _{bcd}
k_a =0$ implies that $I_1 k_e =0$ and hence $I_1 = 0$, where $I_1$ is the Euler density:
$$I_1
\equiv [R^{abcd} R_{abcd} - 4R^{ab} R_{ab} + R^2].
$$

If $R^a \;\! _{bcd;e} k_a \neq 0$, then $\beta =0$ and the metric
is determined up to a constant conformal factor (and the holonomy
type is $R_{15}$).  This is similar to the first case discussed
above, but now some information on the covariant derivative of the
Riemann tensor is necessary (to ensure $R^a \;\! _{bcd;e} k_a \neq
0$).  Hence, first order curvature invariants are needed for the
classification of the spacetime. Since $R^a \;\! _{bcd;e} k_a =0$
implies that $I_2 k_e =0$, where $$I_2 \equiv [R^{abcd;e}
R_{abcd;e} -4R^{ab;c} R_{ab;c} +R^{,a} R_{,a}],$$ it follows that
the invariant $I_2 \neq 0$ implies that $R^a \;\! _{bcd;e} k_a
\neq 0$  in this case.

If $R^a \;\! _{bcd;e} k_a = 0$, then $R^a \;\! _{bcd} k_{a;e} =
0$, and since eqn. (\ref{star}) has a unique solution, $k_a$ is
recurrent. If $k_a$ is null, the spacetime is algebraically
special, and since we assume that the Petrov type is I, this is
not possible. Hence, $k_a$ is (a) timelike (TL) or (b) spacelike
(SL) and is, in fact, covariant constant (CC).

In case (2a), the spacetime admits a TL CC vector field $k_a$.
The holonomy is $R_{13}$, with a TL holonomy invariant subspace
which is non-degenerately reducible, and $M$ is consequently
locally $(1+3)$ decomposable (and static).  There exist local
coordinates (with $k = \frac{\partial}{\partial t}$) such that
the metric is given by
\begin{equation}
-dt^2 + g_{\alpha \beta}(x^\gamma) dx^\alpha dx^\beta \quad (\alpha
=1,2, 3)
\end{equation}
where $g_{\alpha \beta}$ is independent of $t$.  The metric is
unique up to an overall constant scaling and a time translation $t
\to \lambda t$, where $\lambda^2 = 1 + \beta/\alpha$ (reflecting
the non-uniqueness of the TL CC vector up to a constant scaling
$\lambda$).  All of the non-trivial components of the Riemann
tensor and its covariant derivatives are constructed from the 3D
positive definite metric $g_{\alpha \beta}$, and can be classified
by the corresponding 3D Riemann curvature invariants. In this
case (and case ($2b)$) there is an ignorable coordinate and all
invariants  are functions of $3$ independent functions; $R^a \;\!
_{bcd;e}$ must be used to uniquely fix the frame, and hence we
need information from the first order scalar invariants.

In case (2b), the spacetime admits a SL CC vector field $k_a$.
The holonomy is $R_{10}$, there exists a holonomy invariant $SL$
vector $k_a$ which is non-degenerately reducible, and $M$ is this
locally $(3+1)$ decomposable.  Choosing local coordinates in which
the SL CC vector $k = \frac{\partial}{\partial x}$, the metric is given by
\begin{equation}
dx^2 + {g}_{\overline{\alpha}\overline{\beta}}
dx^{\overline{\alpha}} d x^{\overline{\beta}} \quad
(\overline{\alpha} = 0, 2, 3)
\end{equation}
and $g_{\alpha \beta}$ is independent of $x$.  The metric is
unique up to an overall constant conformal factor and a space
translation $x \to \lambda x$ $(\lambda^2 = 1 + \beta/\alpha)$.
Classification now reduces to the classification of the class of
3D Lorentzian spacetimes with Lorentzian metric
${g}_{\overline{\alpha} \overline{\beta}}$ (the subclass such that
(2) is of Petrov type $I$).  We can now iterate the procedure for
3D Lorentzian spacetimes (such that (2) is Petrov type $I$).  In
the degenerate cases in which additional KV are admitted, we will
be led to the locally homogeneous case, and hence the 4D Petrov
type I locally homogeneous spacetimes (which are characterized by
their constant scalar invariants). Indeed, in 3D the Riemann
tensor is completely determined by the Ricci tensor.  There always
exists a frame in which the components of the Ricci tensor are
constants \cite{CSI} and so in this case the 4D spacetime is
Petrov type $I$ and $CH_0$ (curvature homogeneous \cite{Mp}), and
hence generically locally homogeneous.

\section*{Acknowledgments}
We would like to thank Robert Milson for useful comments and questions on our manuscript, and to Lode Wylleman for pointing out a mistake. This work was supported by the
Natural Sciences and Engineering Research Council of Canada. 
\appendix

\section{Notation}\label{notation}
Throughout we have used a Newman-Penrose (NP) tetrad given by
$e_{a}=\{\ell,n,m,\overline{m}\}$ with inner product

\begin{equation}
\eta_{ab}=\begin{bmatrix}
0 & 1 & 0 & 0 \\
1 & 0 & 0 & 0 \\
0 & 0 & 0 & -1 \\
0 & 0 & -1 & 0
\end{bmatrix}
\end{equation}

\noindent and directional derivatives defined by
\begin{eqnarray}
D=\ell^{\mu}\nabla_{\mu}, & \Delta=n^{\mu}\nabla_{\mu}, &
\delta=m^{\mu}\nabla_{\mu} \, .
\end{eqnarray}

Associated with an NP tetrad are the following definitions for the
connection coefficients that appear frequently above
\begin{eqnarray}
\kappa=m^{\mu}D\ell_{\mu}, & \sigma=m^{\mu}\delta\ell_{\mu}, &
\rho=m^{\mu}\overline{\delta}\ell_{\mu}
\end{eqnarray}

\noindent with the remaining ones being similarly defined.  Given
the frame components $R_{abcd}=R_{\alpha\beta\gamma\delta}e_{a}^{\
\alpha}e_{b}^{\ \beta}e_{c}^{\ \gamma}e_{d}^{\ \delta}$, we have
the definitions for the Weyl scalars
\begin{equation}
\begin{array}{ccccc}
\Psi_{0}=-C_{1313}, & \Psi_{1}=-C_{1213}, & \Psi_{2}=-C_{1342}, &
\Psi_{3}=-C_{1242}, & \Psi_{4}=-C_{2424} \nonumber
\end{array}
\end{equation}

\noindent and the Ricci scalars
\begin{equation}
\begin{array}{lll}
\Phi_{00}=\frac{1}{2}R_{11}, & \Phi_{01}=\frac{1}{2}R_{13}, & \Phi_{02}=\frac{1}{2}R_{33} \\
\Phi_{11}=\frac{1}{4}(R_{12}+R_{34}), &
\Phi_{12}=\frac{1}{2}R_{23}, & \Phi_{22}=\frac{1}{2}R_{22} \, .
\end{array}
\end{equation}

Given a covariant tensor $T$ with respect to an NP tetrad (or null
frame), the effect of a boost $\ell \mapsto e^{\lambda}\ell$, $n
\mapsto e^{-\lambda}n$ allows $T$ to be decomposed according to
its boost weight
\begin{equation}
T=\sum_b (T)_{b} \label{bwdecomp}
\end{equation}
where $(T)_{b}$ denotes the boost weight $b$ components of $T$. An
algebraic classification of tensors $T$ has been developed
\cite{class,Milson} which is based on the existence of certain
normal forms of (\ref{bwdecomp}) through   successive application
of null rotations and spin-boost.  In the special case where $T$
is the Weyl tensor in four dimensions, this classification reduces
to the well-known Petrov classification. However, the boost weight
decomposition can be used in the classification of any tensor $T$
in arbitrary dimensions.  As an application, a Riemann tensor of type $G$ has the following decomposition
\begin{equation}
R=(R)_{+2}+(R)_{+1}+(R)_{0}+(R)_{-1}+(R)_{-2}
\end{equation}
\noindent in every null frame.  A Riemann tensor is algebraically special if there exists a frame in which certain boost weight components can be transformed to zero, these are summarized in Table \ref{riemtypes}.

A useful discrete symmetry is the following (orientation-preserving) Lorentz transformation:
\beq
\ell\leftrightarrow n, \quad m\leftrightarrow \bar{m},
\label{dsym}\eeq
which interchanges the boost weights, $(T)_b\leftrightarrow (T)_{-b}$, and makes the replacements
\beq
(\kappa,\sigma,\rho,\tau,\epsilon,\beta)\leftrightarrow -(\nu,\lambda,\mu,\pi,\gamma,\alpha).
\eeq

\begin{table}[h]
\begin{center}
\small{
\begin{tabular}{c|c}
\hline
 & \\
{\bf Riemann type} & {\bf Conditions}  \\
 & \\
\hline
G & --- \\
I & $(R)_{+2}=0$ \\
II & $(R)_{+2}=(R)_{+1}=0$ \\
III & $(R)_{+2}=(R)_{+1}=(R)_{0}=0$ \\
N & $(R)_{+2}=(R)_{+1}=(R)_{0}=(R)_{-1}=0$ \\
D & $(R)_{+2}=(R)_{+1}=(R)_{-1}=(R)_{-2}=0$ \\
O & all vanish (Minkowski space) \\
\hline
\end{tabular}
}
\caption{The relation between Riemann types and the vanishing of boost weight components.  For example, $(R)_{+2}$ corresponds to the frame components $R_{1313},R_{1414},R_{1314}$.}\label{riemtypes}
\end{center}
\end{table}

\section{Some special operators}

Consider the case where we have a tensor $S_{\mu\nu\alpha\beta}$, where
\[ S_{\mu\nu\alpha\beta}=S_{(\mu\nu)(\alpha\beta)}=S_{\alpha\beta\mu\nu},\]
This tensor can be considered as an operator:
\[ {\sf S}=(S^{\mu\nu}_{\phantom{\mu\nu}\alpha\beta}): ~ V\mapsto V, \]
where $V$ is the vector space of symmetric 2-tensors $M^{\mu\nu}$.

Therefore, we can consider the eigentensors of this map in the
standard manner. We can construct a set of projectors $\bot_{A}$
projecting onto each corresponding eigenspace. Assume that
$\bot_1$ is of rank 1 (as an operator). If $M^{\mu\nu}$ is the
corresponding (normalized) eigenvector, this means that
\[ (\bot_1)_{\mu\nu\alpha\beta}=M_{\mu\nu}M_{\alpha\beta}.\]
We can now consider the eigenvectors of ${\sf M}\equiv
M^{\mu}_{~\nu}$.  We are actually not considering the operator
${\sf M}$ itself, but rather ${\bot_1}$. However, $\bot_1$ can
also be considered as an operator:
\[ {\sf P}: N^{\mu\nu}\mapsto M^{\mu}_{~\alpha}M^{\nu}_{~\beta}N^{\alpha\beta}\]
Assume that $v^{\mu}$ and $w^{\nu}$ are eigenvectors of ${\sf M}$
with eigenvalues $\lambda_v$ and $\lambda_w$, respectively. Then,
if $N^{\mu\nu}=v^\mu w^\nu$,
\[ M^{\mu}_{~\alpha}M^{\nu}_{~\beta}N^{\alpha\beta}=\lambda_v\lambda_w N^{\mu\nu},\]
and is therefore an eigenvector of ${\sf P}$ with eigenvalue
$\lambda=\lambda_v\lambda_w$. Clearly, $v^{\mu}w^{\nu}$ has the
same eigenvalue as $w^{\mu}v^{\nu}$, so we will not be able to
distinguish these using projection operators. Furthermore, if
$\lambda_v=\pm \lambda_w$, then $v^{\mu}v^{\nu}$ has the same
eigenvalue as $w^{\mu}w^{\nu}$.

The above construction is useful in several cases. An example that recurs is the case where ${\sf M}$ has two one-dimensional eigenspaces spanned by $v^{\mu}$ and $w^{\mu}$, say. Assume also that $\lambda_v=-\lambda_w$. Then, ${\sf P}$ has two projection operators:
\beq
({\sf P}_1)_{\mu\nu\alpha\beta}&\propto & v_{\mu}w_{\nu}v_{\alpha}w_{\beta}+w_{\mu}v_{\nu}w_{\alpha}v_{\beta}, \\
({\sf P}_2)_{\mu\nu\alpha\beta}&\propto & v_{\mu}v_{\nu}v_{\alpha}v_{\beta}+w_{\mu}w_{\nu}w_{\alpha}w_{\beta},
\eeq
We  see that this is somewhat unfortunate because in spite of the fact that ${\sf M}$ sees the difference between the vectors $v^\mu$ and $w^\nu$, ${\sf P}$ does not. This is related to the fact that for some spacetimes there exists a discrete symmetry which interchanges two spacetimes with identical curvature invariants. Here this manifests itself in that we cannot actually determine which eigenvector correspond to which eigenvalue.

\section{Algebraically special $\nabla C$}\label{nabcinvar}

The relationship between the invariants of the Weyl tensor and the Petrov type
is well known; however, this is not the case for the covariant derivative of the
Weyl tensor.  A similar analysis for $\nabla C$ would require an algebraic
classification based on its boost weight decomposition, and a complete set of
its first order invariants.  We do not attempt to solve this general problem but
rather provide some relations relevant to our paper.  Restricting attention to
four dimensions we define the following tensors
\begin{eqnarray}
\stackrel{2}{T}_{abe}^{\quad\ fgh} & = & C_{abcd;e}C^{cdfg;h} \\
\stackrel{3}{T}_{abe\ \ ijk}^{\quad\ h} & = & \stackrel{2}{T}_{abe}^{\quad\ fgh}C_{fgij;k} \\
\stackrel{3,0}{T}_{abe\ \ ij}^{\quad\ h} & = & \stackrel{2}{T}_{abe}^{\quad\ fgh}C_{fgij} \\
\stackrel{4}{T}_{abe\ \ k}^{\quad\ h \ \ lmn} & = & \stackrel{3}{T}_{abe\ \ ijk}^{\quad\ h}C^{ijlm;n}
\end{eqnarray}
where the number above the tensor refers to the degree in $\nabla C$ or $C$.
All of these tensors are constructed purely from $\nabla C$ with the exception
of $\stackrel{3,0}{T}$ (which involves $C$).  Next, we consider the following first order invariants
\begin{eqnarray}
w_{2,1}=\stackrel{2}{T}_{abe}^{\quad\ abe} & w_{2,2}=\stackrel{2}{T}_{abe}^{\quad\ eab} & w_{2,3}=\stackrel{2}{T}_{a\ \ e\ \ h}^{\ \ e \ \ a \ \ h} \\
w_{3,1}=\stackrel{3,0}{T}_{abe}^{\quad\ eab} & w_{3,2}=\stackrel{3,0}{T}_{abe}^{\quad\ bea} & w_{3,3}=\stackrel{3,0}{T}_{a\ \ e\ \ b}^{\ \ e \ \ b \ \ a} \\
w_{4,1}=\stackrel{4}{T}_{abe\ \ k}^{\quad\ e \ \ abk} & w_{4,2}=\stackrel{4}{T}_{abe\quad\ \ h}^{\quad\ heab} & w_{4,3}=\stackrel{4}{T}_{abe\ \ h}^{\quad\ h \ \ abe} \\
w_{4,4}=\stackrel{4}{T}_{abe\ \ h}^{\quad\ h \ \ eba} & w_{4,5}=\stackrel{4}{T}_{abe\quad\quad m}^{\quad\ abem} & w_{4,6}=\stackrel{4}{T}_{abe\ \ m}^{\quad\ a \ \ ebm} \\
w_{4,7}=\stackrel{4}{T}_{abe\quad\quad m}^{\quad\ maeb}
\end{eqnarray}
in which $w_{n,i}$ denotes the $ith$ invariant of degree $n$ in $\nabla C$ or
$C$.  Since the aligned frames of $\nabla C$ and $C$ need not be the same, the
$w_{3,i}$ are mixed invariants and the remaining invariants are pure $\nabla C$ invariants.
For $\nabla C$ and $C$ algebraically general (type $G$) we obtain the following syzygies
\begin{eqnarray*}
w_{2,1}+2w_{2,2}-2w_{2,3}=0, & w_{3,1}+2w_{3,2}+2w_{3,3}=0, & w_{4,1}-w_{4,3}=0
\end{eqnarray*}
which are  the result of identities, symmetries and dimensionally dependent relations\footnote{Thanks to Jose M. Martin-Garcia for pointing this out to us.} \cite{M-G}.  In subsequent
calculations we always impose these syzygies so that our set reduces to ten
invariants.  Now consider $\nabla C$ of algebraically special type, which is
obtained by setting the minimal number of appropriate boost weight components to
vanish.  We obtain the following results:
\begin{enumerate}

\item If $\nabla C$  is type $II$ or $D$ (i.e., boost weight $+3,+2,+1$

components vanish) then the syzygies  $S_{1}=0$ and $S_{2}=0$ hold.

\item If $\nabla C$ is type $G$ or   type $H$ (i.e., boost weight $+3$ vanish),
 or type $I$ (i.e., boost weight $+3,+2$ vanish) then, in general, $S_{1}\neq 0$ and $S_{2}\neq 0$.
\end{enumerate}
The second statement refers to the most general types of $G$, $H$ or $I$
where no further algebraically special subcases are taken into account.
Below are the expressions for $S_{1}$ and $S_{2}$.  Note that $S_{1}$ is
linear in $w_{4,5}$, and when $S_1 =0$, we use this syzygy in the
derivation\footnote{$w_{4,5}$ does not appear in $S_{2}$.} of $S_{2}$; hence
these two invariant expressions are generally independent.  In type $II$ or $D$ we can regard $S_{1}=0$ as expressing the dependency of $w_{4,5}$ in terms of the other invariants of $S_{1}$. In $S_{2}$ each of the $w_{2,i}$ appear quadratically whereas each of $w_{4,i}$ appear quartically therefore one of these invariants is dependent with respect to the other invariants in $S_{2}$.  Since these syzygies are of degree 8 and 16, and the invariants considered here are of maximum degree 4, one would expect $S_1$ and $S_2$ to attain a simpler form if expressed in terms of higher degree invariants.  These calculations were performed with the aid of \textit{GRTensorII} \cite{grtensor}.
\begin{eqnarray*}
\lefteqn{S_{1}=-14464 w_{2,3} w_{2,2} w_{4,4}-992 w_{2,3} w_{2,2} w_{4,1}+15872 w_{2,3} w_{2,2} w_{4,6}} \\
& &  +7424 w_{2,3} w_{2,2}  w_{4,7}-1600 w_{2,3} w_{2,2} w_{4,2}+1216 w_{2,2}^2 w_{4,2} \\
& & +640 w_{2,3}^2 w_{4,2}+(21504 w_{4,4}+2304 w_{2,2}^2-24576 w_{4,6}-4608 w_{2,3} w_{2,2} \\
& & -6144 w_{4,7}+2304 w_{2,3}^2-768 w_{4,1}+768 w_{4,2}) w_{4,5}-6656 w_{2,2}^2 w_{4,7} \\
& & -154112 w_{4,6} w_{4,4}-8960 w_{2,3}^2 w_{4,6}-464 w_{4,7} w_{4,1}-2272 w_{4,7} w_{4,2} \\
& & -15104 w_{2,2}^2 w_{4,6}+224 w_{2,2}^2 w_{4,1}+7760 w_{4,4} w_{4,2}-2432 w_{4,6} w_{4,1} \\
& & -10240 w_{4,6} w_{4,2}+15232 w_{2,2}^2 w_{4,4}+2680 w_{4,4} w_{4,1}+56320 w_{4,7} w_{4,6} \\
& & +512 w_{2,3}^2 w_{4,1}+6400 w_{2,3}^2 w_{4,4}-41408 w_{4,7} w_{4,4}-36 w_{4,2} w_{4,1} \\
& & -2816 w_{2,3}^2 w_{4,7}+180 w_{4,2}^2+6464 w_{4,7}^2+58832 w_{4,4}^2+94208 w_{4,6}^2-171 w_{4,1}^2
\end{eqnarray*}
\newpage
\begin{eqnarray*}
\lefteqn{S_{2}=-5364449280w_{4,6}^2w_{2,3}w_{2,2}w_{4,2}+2223360w_{2,3}^2w_{4,2}^3} \\ & & -3893760w_{4,2}^3w_{2,2}^2-603625881600w_{4,6}^3w_{4,4}-20101201920w_{4,6}^3w_{2,3}^2+55490641920w_{4,6}^3w_{2,2}^2 \\ & & -22171852800w_{4,7}^3w_{4,4}-4885920w_{4,2}^3w_{4,1}+148414464000w_{4,6}^2w_{4,7}^2+568104468480w_{4,6}^2w_{4,4}^2 \\ & &
-44177817600w_{4,6}^3w_{4,2}-28282060800w_{4,6}^3w_{4,1}+3288600w_{4,2}w_{4,1}^3-81312860160w_{4,7}w_{4,4}^3 \\ & & +68296366080w_{4,7}^2w_{4,4}^2+2340126720w_{2,2}^2w_{4,7}^3+3975480w_{4,2}^2w_{4,1}^2+305528832000w_{4,6}^3w_{4,7} \\ & & +187499520w_{4,6}^2w_{4,1}^2+3116666880w_{4,6}^2w_{4,2}^2-227302871040w_{4,6}w_{4,4}^3+12399045120w_{4,4}^3w_{4,2} \\ & & -80592261120w_{4,4}^3w_{2,2}^2-194522400w_{4,4}^2w_{4,1}^2+2817964800w_{4,4}^3w_{4,1}+2603059200w_{4,4}^3w_{2,3}^2 \\ & & -16178400w_{4,1}^3w_{4,7}+22063680w_{4,1}^3w_{2,2}^2+62340480w_{4,1}^2w_{4,7}^2+41428800w_{4,1}^3w_{4,6} \\ & & -10962000w_{4,1}^3w_{4,4}-1406880w_{4,1}^3w_{2,3}^2+1134028800w_{4,7}^3w_{4,1}-1420185600w_{4,7}^3w_{4,2} \\ & & +31371264000w_{4,7}^3w_{4,6}-612679680w_{4,7}^3w_{2,3}^2-30101760w_{4,2}^3w_{4,7}+71660160w_{4,2}^3w_{4,4} \\ & &  +310187520w_{4,2}^2w_{4,7}^2+1551254400w_{4,2}^2w_{4,4}^2-96145920w_{4,2}^3w_{4,6}+1095120w_{4,2}^4 \\ & & +231211008000w_{4,6}^4+32785562880w_{4,4}^4+496125w_{4,1}^4+2437632000w_{4,7}^4
\\ & & -101231493120w_{4,6}^2w_{2,3}w_{2,2}w_{4,4}+41724149760w_{4,6}^2w_{2,3}w_{2,2}w_{4,7}-10441359360w_{4,6}^2w_{2,3}w_{2,2}w_{4,1} \\ & & +8616960w_{2,3}w_{2,2}w_{4,2}^2w_{4,7}+109117440w_{2,3}w_{2,2}w_{4,2}^2w_{4,6}+61908480w_{2,3}w_{2,2}w_{4,2}^2w_{4,4} \\ & & -266860800w_{2,3}w_{2,2}w_{4,1}^2w_{4,7}+1940244480w_{2,3}w_{2,2}w_{4,1}w_{4,7}^2+2091409920w_{2,3}w_{2,2}w_{4,1}w_{4,4}^2 \\ & &
+414351360w_{2,3}w_{2,2}w_{4,1}^2w_{4,6}+52652160w_{2,3}w_{2,2}w_{4,1}w_{4,2}^2+106024320w_{2,3}w_{2,2}w_{4,1}^2w_{4,4} \\ & &
+50016960w_{2,3}w_{2,2}w_{4,1}^2w_{4,2}+643184640w_{4,2}w_{4,1}w_{2,2}^2w_{4,7}+201784320w_{4,2}w_{4,1}w_{4,6}w_{4,4} \\ & &
-150036480w_{4,2}w_{4,1}w_{2,3}^2w_{4,6}+1132830720w_{4,2}w_{4,1}w_{2,2}^2w_{4,6}-1692518400w_{4,2}w_{4,1}w_{4,7}w_{4,6} \\ & &
-1386869760w_{4,2}w_{4,1}w_{2,2}^2w_{4,4}-74626560w_{4,2}w_{4,1}w_{2,3}^2w_{4,4}+1610933760w_{4,2}w_{4,1}w_{4,7}w_{4,4} \\ & &
+195978240w_{4,2}w_{4,1}w_{2,3}^2w_{4,7}-4386816000w_{2,3}^2w_{4,2}w_{4,6}w_{4,4}+2208890880w_{2,3}^2w_{4,2}w_{4,7}w_{4,6} \\ & &
-1048596480w_{2,3}^2w_{4,2}w_{4,7}w_{4,4}+26542080000w_{4,6}w_{4,4}^2w_{2,3}w_{2,2}-141363118080w_{4,6}w_{4,4}w_{2,2}^2w_{4,7} \\ & &
+3734138880w_{4,6}w_{4,4}w_{4,7}w_{4,1}+60439633920w_{4,6}w_{4,4}w_{4,7}w_{4,2}+19318947840w_{4,6}w_{4,4}w_{2,2}^2w_{4,1} \\ & &
+18244730880w_{4,6}w_{4,4}w_{2,2}^2w_{4,2}-2918891520w_{4,6}w_{4,4}w_{2,3}^2w_{4,1}+28718530560w_{4,6}w_{4,4}w_{2,3}^2w_{4,7} \\ & &
-14858588160w_{4,7}w_{4,4}^2w_{2,3}w_{2,2}+3924910080w_{4,7}^2w_{4,4}w_{2,3}w_{2,2}+9899274240w_{4,7}w_{4,4}w_{2,2}^2w_{4,1} \\ & &
+6060810240w_{4,7}w_{4,4}w_{2,2}^2w_{4,2}+304588800w_{4,7}w_{4,4}w_{2,3}^2w_{4,1}+1175224320w_{2,3}^2w_{4,6}w_{4,7}w_{4,1} \\ & &
-7813693440w_{2,2}^2w_{4,7}w_{4,6}w_{4,1}-5205196800w_{2,2}^2w_{4,7}w_{4,6}w_{4,2}+2413071360w_{4,4}^2w_{2,3}w_{2,2}w_{4,2} \\ & &
+13086720w_{2,3}w_{2,2}w_{4,7}^2w_{4,2}+5733089280w_{2,3}w_{2,2}w_{4,7}^2w_{4,6}-1048320w_{2,3}w_{2,2}w_{4,2}^3 \\ & &
-312531840w_{4,2}w_{4,1}w_{4,4}^2-25586565120w_{4,6}^2w_{2,3}^2w_{4,7}-21538897920w_{4,6}w_{4,4}^2w_{2,3}^2 \\ & &
-204650496000w_{4,6}w_{4,4}w_{4,7}^2-17712000w_{4,2}^2w_{4,1}w_{2,3}^2+78222827520w_{4,7}w_{4,4}^2w_{2,2}^2 \\ & &
-84564000w_{4,2}w_{4,1}^2w_{4,4}+824785920w_{2,3}^2w_{4,2}w_{4,4}^2+426528000w_{4,6}w_{4,4}w_{4,1}^2 \\ & & -7596933120w_{4,6}^2w_{2,2}^2w_{4,2}
-202144481280w_{4,6}^2w_{2,2}^2w_{4,4}+152928000w_{4,2}^2w_{4,1}w_{4,6} \\ & & +90594754560w_{4,6}^2w_{4,4}w_{4,2}
-158883840w_{2,3}^2w_{4,2}^2w_{4,6}-4459392000w_{4,6}w_{4,4}w_{4,2}^2 \\ & & +2974187520w_{4,6}^2w_{2,3}^2w_{4,1}-33359040w_{4,2}w_{4,1}^2w_{4,7}
-46275840w_{4,2}^2w_{4,1}w_{2,2}^2 \\ & & +91186560w_{4,2}^2w_{4,1}w_{4,7}+277770240w_{2,3}^2w_{4,2}w_{4,7}^2+59454259200w_{4,6}^2w_{2,2}^2w_{4,7} \\ & &
+202020480w_{4,2}w_{4,1}^2w_{4,6}+38804520960w_{4,6}^2w_{2,3}^2w_{4,4}-560862720w_{4,2}w_{4,1}w_{4,7}^2 \\ & & +78681600w_{2,3}^2w_{4,2}^2w_{4,4}-4388981760w_{4,7}^2w_{4,4}w_{4,1}-20584074240w_{4,7}w_{4,4}^2w_{4,2}
\end{eqnarray*}
\begin{eqnarray*}
 & & -23960862720w_{4,7}^2w_{4,4}w_{2,2}^2
+396922429440w_{4,6}w_{4,4}^2w_{4,7}-43977600w_{4,2}w_{4,1}^2w_{2,2}^2 \\ & & +456929280w_{4,6}^2w_{4,2}w_{4,1}+664450560w_{4,7}w_{4,4}^2w_{4,1} +3913482240w_{4,6}^2w_{2,3}^2w_{4,2} \\ & & +495043200w_{4,7}w_{4,4}w_{4,1}^2-20418462720w_{4,6}w_{4,4}^2w_{4,1}
-614758809600w_{4,6}^2w_{4,7}w_{4,4} \\ & & -143570880w_{4,2}^2w_{4,1}w_{4,4}-1454814720w_{4,7}w_{4,4}w_{4,2}^2 -5780275200w_{4,6}^2w_{2,2}^2w_{4,1} \\ & & -12847680w_{4,2}w_{4,1}^2w_{2,3}^2+44401582080w_{4,6}^2w_{4,4}w_{4,1}
+9839646720w_{4,7}^2w_{4,4}w_{4,2} \\ & & -59617105920w_{4,6}w_{4,4}^2w_{4,2}-7641907200w_{4,6}^2w_{4,7}w_{4,1}+205920w_{2,3}w_{2,2}w_{4,1}^3 \\ & & -42624000w_{2,3}^2w_{4,2}^2w_{4,7}-43020288000w_{4,6}^2w_{4,7}w_{4,2}+226496839680w_{4,6}w_{4,4}^2w_{2,2}^2 \\ & & -5489233920w_{4,7}w_{4,4}^2w_{2,3}^2+3500236800w_{4,7}^2w_{4,4}w_{2,3}^2-7591034880w_{2,3}^2w_{4,6}w_{4,7}^2 \\ & &
-91699200w_{2,3}^2w_{4,6}w_{4,1}^2+1986324480w_{4,2}^2w_{4,7}w_{4,6}-377994240w_{4,2}^2w_{2,2}^2w_{4,4} \\ & & +10468362240w_{4,4}^3w_{2,3}w_{2,2}-10165155840w_{4,4}^2w_{2,2}^2w_{4,2}-856350720w_{2,2}^2w_{4,7}^2w_{4,2} \\ & & +20755906560w_{2,2}^2w_{4,7}^2w_{4,6}+101836800w_{2,2}^2w_{4,7}w_{4,2}^2+237242880w_{2,2}^2w_{4,7}w_{4,1}^2  \\ & & +321085440w_{2,2}^2w_{4,6}w_{4,2}^2-167454720w_{2,2}^2w_{4,6}w_{4,1}^2-2234142720w_{2,2}^2w_{4,7}^2w_{4,1} \\ & &
-152616960w_{2,3}w_{2,2}w_{4,7}^3-228591360w_{2,2}^2w_{4,4}w_{4,1}^2-12494730240w_{2,2}^2w_{4,4}^2w_{4,1} \\ & &
-13230720w_{2,3}^2w_{4,4}w_{4,1}^2+260743680w_{2,3}^2w_{4,4}^2w_{4,1}+68417280w_{2,3}^2w_{4,7}w_{4,1}^2 \\ & & +65399685120w_{4,6}^3w_{2,3}w_{2,2}
+6782976000w_{2,3}w_{2,2}w_{4,1}w_{4,6}w_{4,4}-534067200w_{2,3}^2w_{4,7}^2w_{4,1} \\ & & -1123038720w_{4,1}^2w_{4,7}w_{4,6} +4393267200w_{4,7}^2w_{4,6}w_{4,1}-13674700800w_{4,7}^2w_{4,6}w_{4,2} \\ & & -644705280w_{2,3}w_{2,2}w_{4,1}w_{4,7}w_{4,2} +780894720w_{2,3}w_{2,2}w_{4,1}w_{4,4}w_{4,2} \\ & & -204963840w_{2,3}w_{2,2}w_{4,1}w_{4,6}w_{4,2} +14745600w_{2,3}w_{2,2}w_{4,1}w_{4,7}w_{4,6} \\ & & -4408104960w_{2,3}w_{2,2}w_{4,1}w_{4,7}w_{4,4} +1975910400w_{4,6}w_{4,4}w_{2,3}w_{2,2}w_{4,2} \\ & & -19641139200w_{4,6}w_{4,4}w_{2,3}w_{2,2}w_{4,7} -1053757440w_{4,7}w_{4,4}w_{2,3}w_{2,2}w_{4,2} \\ & & -1527644160w_{2,3}w_{2,2}w_{4,2}w_{4,7}w_{4,6}
\end{eqnarray*}

\section{Curvature}
Let $M$ be a 4-dimensional smooth connected Hausdorff manifold
admitting a global smooth Lorentz
metric $h$ with
associated curvature tensor $R$. It will be convenient to
describe a simple algebraic classification of $R$  according to its rank
(relative to
$h$). This classification is easily described geometrically and is
a pointwise classification \cite{HallBook}.

A skew-symmetric tensor $F$ of type $(0,2)$ or $(2,0)$ at $m\in T_mM$ is
called a {\em bivector}. If $F(\neq 0)$ is such a bivector, the
rank of any of its (component) matrices is either two or four. In
the former case, one may write (e.g. in the $(2,0)$ case)
$F^{ab}=2r^{[a}s^{b]}$ for $r, s\in T_mM$ (or alternatively,
$F=r\wedge s$) and $F$ is called {\em simple}, with the
2-dimensional subspace (2-space) of $T_mM$ spanned by $r, s$
referred to as the {\em blade} of $F$. In the latter case, $F$ is
called {\em non-simple}.

The metric $h(m)$ converts $T_mM$ into a Lorentz inner product
space and thus it makes sense to refer to vectors in $T_mM$ and
covectors in the cotangent space $T_m^{*}M$ to $M$ at $m$ (using
$h(m)$ to give a unique isomorphism $T_mM\leftrightarrow
T_m^{*}M$, that is, to raise and lower tensor indices) as being
{\em timelike}, {\em spacelike}, {\em null} or {\em orthogonal},
using the signature $(-,+,+,+)$. The same applies to 1-dimensional
subspaces ({\em directions}) and 2- and 3-dimensional subspaces of
$T_mM$ or $T^{*}_mM$. A {\em simple} bivector at $m$ is then
called {\em timelike} (respectively, {\em spacelike} or {\em
null}) if its blade at $m$ is a timelike (respectively a
spacelike or null) 2-space at $m$. A {\em non-simple} bivector
$F$ at $m$ may be shown to uniquely determine an orthogonal
pair of 2-spaces at $m$, one spacelike and one timelike, and
which are referred to as the {\em canonical pair of blades} of
$F$. A tetrad $(l,n,x,y)$ of members of $T_mM$ is called a null
tetrad at $m$ if the only non-vanishing inner products between
its members at $m$ are $h(l,n)=h(x,x)=h(y,y)=1$. Thus $l$ and $n$
are null.

\subsection{Classification}

Define a linear map $f$ from the 6-dimensional vector space of
type $(2,0)$ bivectors at $m$ into the vector space of type
$(1,1)$ tensors at $m$ by $f:F^{ab}\rightarrow R^a_{\
bcd}F^{cd}$. The condition (2) shows that if a tensor $T$ is in
the range of $f$ then
\begin{equation}
h_{ae}T^e_{\ b}+h_{be}T^e_{\ a}=0 \ \ \ \ (\Rightarrow
T_{ab}=-T_{ba}, \ \  T_{ab}=h_{ae}T^e_{\ b})
\end{equation}
and so $T$ can be regarded as a member of the matrix
representation of the Lie algebra of the pseudo-orthogonal
(Lorentz) group of $h(m)$. Using $f$ one can divide the curvature
tensor $R(m)$ into five classes.

\begin{description}
\item[Class $A$] This is the most general curvature class and the curvature
will be said to be of (curvature) class $A$ at $m\in M$ if it is
not in any of the classes $B$, $C$, $D$ or $O$ below.
\item[Class $B$] The curvature tensor is said to be of (curvature)
class $B$ at $m\in M$ if the range of
$f$ is 2-dimensional and consists of all linear combinations of
type $(1,1)$ tensors $F$ and $G$ where $F^a_{\ b}=x^ay_b-y^ax_b$
and $G^a_{\ b}=l^an_b-n^al_b$ with $l,n,x,y$ a null tetrad at
$m$. The curvature tensor at $m$ can then be written as
\begin{equation}
  R_{abcd}\equiv h_{ae}R^e_{\ bcd}=\frac{\alpha}{2} F_{ab}F_{cd}-\frac{\beta}{2} G_{ab}G_{cd}
\end{equation}
where $\alpha,\beta\in\mathbb{R}$, $\alpha\neq 0\neq\beta$.
\item[Class $C$] The curvature tensor is said to be of (curvature)
class $C$ at $m\in M$ if the range of $f$ is 2- or 3-dimensional
and if there exists $0\neq k\in T_mM$ such that each of the type
$(1,1)$ tensors in the range of $f$ contains $k$ in its kernel
(i.e. each of their matrix representations $F$ satisfies $F^a_{\
b}k^b=0$).
\item[Class $D$] The curvature tensor is said to be of
(curvature) class $D$ at $m\in M$ if the range of $f$ is
1-dimensional. It follows that the curvature components satisfy
$R_{abcd}=\lambda F_{ab}F_{cd}$ at $m$ $(0\neq\lambda\in
\mathbb{R})$ for some bivector $F$ at $m$ which then satisfies
$F_{a[b}F_{cd]}=0$ and is thus simple.
\item[Class $O$] The
curvature tensor is said to be of (curvature) class $O$ at $m\in
M$ if it vanishes at $m$.
\end{description}
The following results are useful \cite{HallBook}:
\begin{enumerate}
\item For the classes $A$ and $B$ there does {\bf not} exist
$0\neq k\in T_mM$ such that $F^a_{\ b}k^b=0$ for {\em every} $F$
in the range of $f$.
\item For class $A$, the range of $f$ has
dimension at least two and if this dimension is four or more the
class is necessarily $A$.
\item The vector $k$ in the definition of class $C$ is
unique up to a scaling.
\item For the classes $A$ and $B$ there
does {\bf not} exist $0\neq k\in T_mM$ such that $R^a_{\
bcd}k^d=0$, whereas this equation has exactly one independent
solution for class $C$ and two for class $D$.
\item The five classes $A$, $B$, $C$, $D$ and $O$ are mutually
exclusive and exhaustive for the curvature tensor at $m$. If the
curvature class is the same at each $m\in M$ then $M$ will be
said to be of that class.
\end{enumerate}

\subsection{Properties}

Suppose that the components of the Riemann tensor $R^a_{bcd}$ are
given in a coordinate domain $U$ with metric $h$. Suppose that
$h'$ is another metric with the same components $R^a_{bcd}$. It
follows from \cite{HALL} that:

\begin{flushright}
$ \begin{array}{llcr}
\mathrm{Class\ A}\hspace{1cm} & h'_{ab}=\alpha h_{ab} & \hspace{1cm} & (42a)\hspace{-2mm} \\
\mathrm{Class\ B}\hspace{1cm} & h'_{ab}=\alpha h_{ab}+2\beta
l_{(a}n_{b)}=(\alpha+\beta)h_{ab}-\beta(x_ax_b+y_ay_b) & \hspace{1cm}& (42b)\hspace{-2mm} \\
\mathrm{Class\ C}\hspace{1cm} & h'_{ab}=\alpha h_{ab}+\beta k_ak_b
& \hspace{1cm}&
(42c)\hspace{-2mm} \\
\mathrm{Class\ D}\hspace{1cm} & h'_{ab}=\alpha h_{ab}+\beta
r_ar_b+\gamma
s_as_b+2\delta r_{(a}s_{b)} & \hspace{1cm}& (42d)\hspace{-2mm} \\
\end{array}$
\end{flushright} \addtocounter{equation}{1} where $\alpha, \beta,
\gamma, \delta\in\mathbb{R}$.

Note that $h_{ab;c}=h_{ab}w_c$ for some smooth 1-form $w$ on the
open subset $A$. Using condition (2) above and the Ricci identity we get
$h_{ab;[cd]}=0$, which implies  $h_{ab}w_{[c;d]}=0$; thus
$w_{[c;d]}=0$ and so $w_a$ is locally a gradient. Hence, for each
$m\in A$, there is an open neighborhood $W$ of $m$ on which
$w_a=w_{,c}$ for some smooth function $w$. Then on $W$,
$g_{ab}=e^{-w}h_{ab}$ satisfies $g_{ab;c}=0$. Further, if $g'$ is
any other local metric defined on some neighborhood $W'$ of $m$
and compatible with $\Gamma$ then $g'$ satisfies condition (2) on $W'$ and
hence, on $W\cap W'$, $g'=\phi g$ for some positive smooth
function $\phi$. From this and the result $g'_{ab;c}=0$ it follows
that $g'$ is a constant multiple of $g$ on $W\cap W'$.

\end{document}